\documentclass[a4paper,10pt]{article}
\usepackage{amsmath}
\usepackage{amsfonts}
\usepackage{amssymb}
\usepackage{amsthm}
\usepackage{graphicx}
\usepackage{authblk}
\usepackage{stmaryrd}
\usepackage{fullpage}
\usepackage{color}
\usepackage{tikz}
\usepackage{tikz-cd}
\usepackage{hyperref}
\usepackage[colorinlistoftodos]{todonotes}

\newtheorem{theorem}{Theorem}[section]
\newtheorem{proposition}[theorem]{Proposition}
\newtheorem{lemma}[theorem]{Lemma}

\theoremstyle{definition}

\newtheorem{example}[theorem]{Example}
\newtheorem{remark}[theorem]{Remark}

\newcommand{\id}{\textit{id}}
\newcommand{\Pow}{\mathcal{P}}
\newcommand{\ad}{\textit{adom}}

\newcommand{\dom}{\textit{dom}}
\newcommand{\cod}{\textit{cod}}
\newcommand{\op}{\textit{op}}

\newcommand{\Dom}{\textit{Dom}}
\newcommand{\Cod}{\textit{Cod}}

\newcommand{\mcomp}{\odot}
\newcommand{\Sup}{\bigvee}
\newcommand{\Inf}{\bigwedge}
\renewcommand{\sup}{\vee}
\renewcommand{\inf}{\wedge}
\newcommand*\sscirc{\cdot} 

\newcommand{\rtof}{\mathcal{F}}
\newcommand{\ftor}{\mathcal{R}}

\definecolor{ufcolor}{rgb}{1,0.5,0.5}
\definecolor{gscolor}{cmyk}{1,0,0,0}
\definecolor{cjcolor}{cmyk}{0,1,0,0}
\definecolor{kzcolor}{cmyk}{0,0,1,0}

\begin{document}

\title{$\ell r$-Multisemigroups and Modal Convolution Algebras}

\author[1]{Uli Fahrenberg}
\author[2]{Christian Johansen} 
\author[3]{Georg Struth} 
\author[4]{Krzysztof Ziemia\'nski}
\affil[1]{{\'E}cole Polytechnique, Palaiseau, France}
\affil[2]{Norwegian University of Science and Technology, Norway} 
\affil[3]{University of Sheffield, UK }
\affil[4]{University of Warsaw, Poland}

\maketitle

\begin{abstract}
  We show how modal quantales arise as convolution algebras $Q^X$ of
  functions from $\ell r$-multisemi\-groups $X$, that is,
  multisemigroups with a source map $\ell$ and a target map $r$, into
  modal quantales $Q$, which can be seen as weight or value
  algebras. In the tradition of boolean algebras with operators we
  study modal correspondences between algebraic laws in $X$, $Q$ and
  $Q^X$.  The class of $\ell r$-multisemigroups introduced in this
  article generalises Schweizer and Sklar's function systems and
  object-free categories to a setting isomorphic to algebras of
  ternary relations as used in boolean algebras with operators and in
  substructural logics. Our results provide a generic construction
  recipe for weighted modal quantales from such multisemigroups. This
  is illustrated by many examples, ranging from modal algebras of
  weighted relations as used in fuzzy mathematics, category quantales
  in the tradition of category algebras or group rings, incidence
  algebras over partial orders, discrete and continuous weighted path
  algebras, weighted languages of pomsets with interfaces, and
  weighted languages associated with presimplicial and precubical
  sets. We also discuss how these results can be combined with
  previous ones for concurrent quantales and generalised to a setting
  that supports reasoning with stochastic matrices or probabilistic
  predicate transformers.

  \vspace{\baselineskip}

  \textbf{Keywords:} multisemigroups, quantales, convolution algebras,
  (object-free) categories, modal algebras, quantitative program
  verification.
\end{abstract}



\section{Introduction}\label{S:introduction}

Convolution is an ubiquitous operation in mathematics and
computing. Let, for instance, $(\Sigma^\ast,\odot,\varepsilon)$ be the
free monoid over the finite alphabet $\Sigma$ and $(S,+,\cdot,0,1)$ a
semiring. The convolution of two functions $f,g:\Sigma^\ast \to S$ is
then defined as
\begin{equation*}
  (f\ast g) (x) = \sum_{x=y\odot z}f(y)\cdot g(z),
\end{equation*}
where $\sum$ represents finitary addition in $S$.  The functions $f$
and $g$ associate a value or weight in $S$ with any word in
$\Sigma^\ast$.  The weight of the convolution $f\ast g$ and word $x$
is thus computed by splitting $x$ into all possible words $y$ and $z$
such that $x=y\odot z$, multiplying their weights $f(y)$ and $g(z)$ in
$S$, and then adding up the results for all appropriate $y$ and $z$.
A finite addition suffices in this example because only finitely many
$y$ and $z$ satisfy $x=y\odot z$ for any $x$.  If $S=\{0,1\}=2$ is the
semiring of booleans (with $\max$ as $+$ and $\min$ as $\cdot$),
$f:\Sigma^\ast\to 2$ becomes a characteristic function for a set;
$f(x)$ can be read as $x\in f$ and convolution is language product.
The generalisation to arbitrary semirings thus yields weighted
languages~\cite{book/DrosteKV09}, and it can be shown that convolution
algebras $S^{\Sigma^\ast}$ of weighted languages form again semirings
with convolution as multiplication.

This data can be changed in various ways. We are interested in the
case where $\Sigma^\ast$ is generalised to a multisemigroup
$(X,\odot)$ with multioperation $\odot:X\times X \to \Pow X$
satisfying a suitable associativity law~\cite{KudryavtsevaM15}, or to
a multimonoid that can have many units. Multisemigroups specialise to
partial semigroups when $|x\odot y| \le 1$ for all $x,y\in X$ and to
semigroups when $|x\odot y|=1$. We also replace $S$ by a quantale $Q$,
a complete lattice with an additional monoidal composition and unit
satisfying certain sup-preservation laws.  The existence of arbitrary
sups in $Q$ then compensates for the lack of finite decompositions in
$X$, which were previously available in $\Sigma^\ast$.  For functions
$f,g:X\to Q$, convolution now becomes
\begin{equation*}
  (f\ast g)(x) = \Sup_{x\in y\odot z}f(y)\cdot g(z),
\end{equation*}
where $\Sup$ indicates a supremum and $\cdot$ the multiplication in
$Q$. This time, the convolution algebra $Q^X$ can be equipped with a
quantale structure~\cite{DongolHS21}. For the quantale $Q=2$ of
booleans it becomes a powerset quantale. Prime examples for this
construction are the convolution quantales of binary relations and of
quantale-valued binary relations~\cite{Goguen67}, as well as
boolean-valued and quantale-valued matrices. These lift from instances
of multimonoids known as pair groupoids.

Yet the relationship between multimonoids, quantales and convolution
quantales is not just a lifting. Multimonoids are isomorphic to
relational monoids formed by ternary relations with suitable monoidal
laws, where $R(x,y,z)$ holds if and only if $x\in y\odot z$. The
construction of convolution algebras can thus be seen in light of
J\'onsson and Tarski's duality between boolean algebras with $n$-ary
operators and $n+1$-ary relations~\cite{JonssonT51,Goldblatt89}, and
in particular of the modal correspondences between these algebras and
relations.  Convolution can be seen as a generalised binary modality
on $Q^X$ and $X$ as a ternary frame, so we may ask about
correspondences between equations in the algebras $X$ and $Q$ and
$Q^X$.

Such correspondences between monoidal properties in $X$ and quantalic
properties in $Q$ and $Q^X$ have already been studied and adapted to
situations where $Q$ is merely a semiring~\cite{CranchDS20a}, for
$Q=2$, they are well known from substructural logics. They have
also been extended to correspondences between concurrent relational
monoids and the concurrent semirings and quantales previously
introduced by Hoare and coworkers~\cite{HoareMSW11}, where two
quantalic compositions are present, which interact via a weak
interchange law, as known from higher categories.

Here, our main motivation lies in the study of correspondences between
the source and target structures or the domain and codomain structures
that are present in multimonoids and many convolution quantales,
including powerset quantales. On the one hand, every element of a
multimonoid has a fortiori a unique left and a unique right unit,
which can be captured using source and target maps like in a
category~\cite{CranchDS20}, and every small category can of course be
modelled in object-free style as a partial semigroup equipped with
such maps~\cite{MacLane98}. On the other hand, quantales of binary
relations, for instance, have a non-trivial domain and codomain
structure, which has led to more abstract definitions of modal
semirings and
quantales~\cite{DesharnaisMS06,DesharnaisS11,FahrenbergJSZ20}. Yet
what is the precise correspondence between the source and target
structure of the pair groupoid, say, and these relational domain and
codomain operations? And how does this correspondence generalise to
source and target maps in arbitrary multisemigroups and domain and
codomain operations in arbitrary value and convolution quantales?

As a first step towards answers, we introduce
$\ell r$-multisemigroups: multisemigroups $(X,\odot)$ equipped with
two operations $\ell,r:X\to X$ inspired by similar ones that appear in
Schweizer and Sklar's function systems~\cite{SchweizerS67} and by the
source and target maps of object-free categories~\cite{MacLane98}. We
study the basic algebra of $\ell r$-multisemigroups and present a
series of examples, including categories and non-categories.  Most of
these results have been formalised using the Isabelle/HOL proof
assistant.\footnote{
  \url{https://github.com/gstruth/lr-multisemigroups}} It turns out,
in particular, that the locality laws
$\ell(x\odot \ell(y)) = \ell(x\odot y)$ and
$r(r(x)\odot y) = r(x\odot y)$ of $\ell r$-multisemigroups, which have
previously been studied in the context of modal semigroups, semirings
and
quantales~\cite{DesharnaisJS09,DesharnaisMS06,DesharnaisS11,FahrenbergJSZ20},
are equivalent to the typical composition pattern of categories,
namely that $x\mcomp y$ is defined (and hence not $\emptyset$) if and
only if $r(x)= \ell(y)$ (for categories, the order of composition
needs to be swapped).  Indeed, partial $\ell r$-semigroups that
satisfy locality \emph{are} (small) categories.

As a second step, we generalise the standard definitions of domain and
codomain used in modal semirings and quantales to ones suitable for
convolution quantales. When relations are represented at $2$-valued
matrices, for instance, domain elements in $2^X$ are diagonal matrices
in which the value of each diagonal element is $1$ if there is a $1$
in the respective row of the matrix, and $0$ otherwise. More
generally, for $Q$-valued relations represented as $Q$-valued
matrices, a domain element in $Q^X$ can be seen as a diagonal matrix
in which the value of each diagonal element is the domain of the
supremum of the respective row values of the matrix, taken in $Q$.

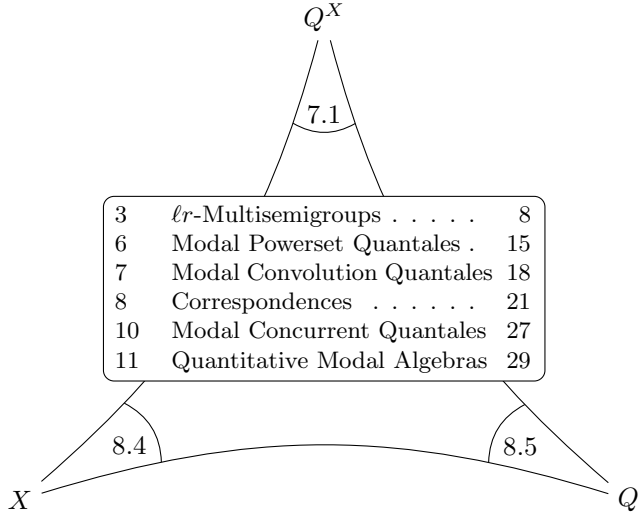
\begin{figure}[tp]
  \centering
  \begin{tikzpicture}[x=8cm, y=8cm]
    \node (X) at (0,0) {$X$};
    \node (Q) at (1,0) {$Q$};
    \node (R) at (.5,.8) {$Q^X$};
    \path (X) edge[bend left=4ex] node[coordinate, pos=.2] (QX) {} node[coordinate, pos=.8] (XQ) {} (Q);
    \path (X) edge[bend right=4ex] node[coordinate, pos=.2] (RX) {} node[coordinate, pos=.85] (XR) {} (R);
    \path (Q) edge[bend left=4ex] node[coordinate, pos=.2] (RQ) {} node[coordinate, pos=.85] (QR) {} (R);
    \path (XR) edge[bend right] node[above] {\ref{P:lr-lifting}}(QR);
    \path (QX) edge[bend right] node[below] {\ref{P:lr-correspondence2-thm}\;\;\;\;\;\;} (RX);
    \path (XQ) edge[bend left] node[below] {\;\;\;\;\;\;\ref{T:correspondence-X+QX->Q}} (RQ);
    \node[rectangle, draw, rounded corners, fill=white] at (.5,.35) {%
      \begin{minipage}{.35\linewidth}
        \hspace*{-3ex}%
        \begin{minipage}{1.07\linewidth}
        \raggedright\small
        \makeatletter\let\l@section\l@subsection\makeatother
\contentsline {section}{\numberline {3}$\ell r$-Multisemigroups}{8}{section.3}%
\contentsline {section}{\numberline {6}Modal Powerset Quantales}{15}{section.6}%
\contentsline {section}{\numberline {7}Modal Convolution Quantales}{18}{section.7}%
\contentsline {section}{\numberline {8}Correspondences}{21}{section.8}%
\contentsline {section}{\numberline {10}Modal Concurrent Quantales}{27}{section.10}%
\contentsline {section}{\numberline {11}Quantitative Modal Algebras}{29}{section.11}%
        \end{minipage}
      \end{minipage}
    };
  \end{tikzpicture}
  \caption{Triangle of correspondences between identities in families
    of $\ell r$-multisemigroups $X$, quantales $Q$, and convolution
    algebras $Q^X$. The box contains a miniature table of content
    pointing to the sections with the main contributions.}
  \label{fig:correspondence-triangle}
\end{figure}

Equipped with $\ell r$-multisemigroups and the generalised domain and
codomain operations for $Q^X$, we prove the main results in this
article: correspondences between equational properties of source and
target operations in $\ell r$-multisemigroups and those of domain and
codomain operations in the modal quantales $Q$ and $Q^X$.  The
resulting triangle, with links to the main theorems that capture them,
is shown in Figure~\ref{fig:correspondence-triangle}. We develop these
in full generality, but keep an eye on the $Q=2$ case which has been
formalised with Isabelle.  From the point of view of boolean algebras
with operators, this yields a multimodal setting where the quantalic
composition in $Q^X$ is a generalised binary modality associated with
a ternary frame in $X$, whereas the domain and codomain operations in
$Q^X$ are generalised unary modalities that can be associated with
binary frames in $X$ based on $\ell$ and $r$.

These results show how the equational axioms for modal semirings and
quantales, as powerset or proper convolution algebras, arise from the
much simpler ones of $\ell r$-multisemigroups and how in particular
the composition pattern of categories translates into the locality
laws of modal convolution algebras. More generally, the lifting from
$X$ and $Q$ to $Q^X$ yields a generic construction recipe for modal
quantales. It shows that any category, for instance, can be lifted in
such a way, but we present other examples such as generalised effect
or shuffle algebras, where locality laws are absent or which are
proper multisemigroups.  Constructing variants of modal quantales is
interesting for program verification because these algebras are
related to dynamic logics, Hoare logics and algebras of predicate
transformers. Conversely, by these correspondences, before
constructing a modal quantale one should ask what the underlying
$\ell r$-multisemigroup might be. The benefit is that the
$\ell r$-multisemigroup laws are much easier to check than those of
modal quantales and in particular proof assistants can benefit from
our generic construction.

The relevance and potential of the construction is underpinned by many
examples. Most of the $\ell r$-multisemigroups presented are
categories, for instance categories of paths in quivers or Moore paths
in topology, pair groupoids, segments and intervals of posets or
pomsets, simplices and cubical sets with interfaces, but some others
are neither local nor partial. For all of them we get modal
convolution quantales for free, and even modal convolution semirings,
if the underlying $\ell r$-multisemigroup is finitely decomposable
like $\Sigma^\ast$ above. This is in particular the case for Rota's
incidence algebras constructed over locally finite
posets~\cite{Rota64}. Well known constructions from algebra, such as
matrix rings, group rings, or category algebras, are closely related,
yet typically use rings instead of semirings or quantales as
value-algebras.

Beyond the results mentioned, we sketch a combination of the
correspondence results for concurrent quantales and those for modal
quantales, with obvious relevance to true concurrency semantics in
computing, for instance based on pomsets or digraphs with
interfaces~\cite{Winkowski77,FioreD13,FahrenbergJST20}.  Details are
left to a successor article. We also discuss the relationship between
the domain and codomain operations in modal convolution quantales, and
the diamond and box operators that can be defined using them. This
opens the door to convolution algebras of quantitative predicate
transformer semantics for programs, including probabilistic or fuzzy
ones, just like in the qualitative case, using for instance the
Lawvere quantale or related quantales over the unit interval as
instances of $Q$.  Yet the definitions of domain and codomain in $Q^X$
mentioned, despite being natural for correspondence results, turn out
to be too strict for dealing with stochastic matrices or Markov
chains. We outline how more liberal definitions of domain and codomain
in $Q^X$ can be used instead, yet leave details once again to future
work.

The remainder of this article is organised as
follows. Section~\ref{S:relational-semigroups} outlines the basics of
multisemigroups and multimonoids with many units, relates them with
object-free categories and lists some
examples. Section~\ref{S:lr-semigroups} introduces the first
contribution of this article: $\ell r$-multisemigroups, relates them
with Schweizer and Sklar's function systems~\cite{SchweizerS67}, with
domain semigroups~\cite{DesharnaisJS09}, with multimonoids and an
alternative axiomatisation of object-free categories using source and
target maps. Once again, examples are
listed. Section~\ref{S:quantales} recalls notions of quantales and the
construction of convolution quantales from relational and
multisemigroups~\cite{DongolHS21}.  Section~\ref{S:modal-quantales}
recalls notions of modal quantales~\cite{FahrenbergJSZ20} and
discussed their structure. As a warm-up for more general
constructions, Section~\ref{S:lr-pow} shows how modal powerset
quantales can be lifted from $\ell r$-multisemigroups; proofs can be
found in Appendix~\ref{A:three}. Section~\ref{S:conv-alg} then
presents the first main result of this article: the construction of
modal convolution quantales $Q^X$ from $\ell r$-semigroups $X$ and
modal value quantales $Q$, using appropriate definitions of domain and
codomain on convolution algebras. It also discusses finitely
decomposable $\ell r$-multisemigroups, where value quantales can be
replaced by value semirings and similar structures.
Section~\ref{S:correspondences} contains the remaining main results,
which complete the triangle of correspondences shown in
Figure~\ref{fig:correspondence-triangle}.  A large list of examples is
discussed in Section~\ref{S:examples}. These range from weighted path,
interval and pomset languages over $\ell r$-semigroups that arise in
algebraic topology to the weighted assertion quantales of separation
logic. A combination of the results for the algebras in this article
with previous ones suitable for concurrency~\cite{CranchDS20a} is
outlined in Section~\ref{S:modal-chantales}. The final technical
section, Section~\ref{S:dualities}, discusses the modal box and
diamond operators that arise in convolution quantales and some
selected models, before Section~\ref{S:conclusion} concludes the main
part of the article. Finally, as many algebraic definitions are
scattered across this article, we summarise the most important ones in
Appendix~\ref{A:one}. Appendix~\ref{A:two} explains the equivalence of
multisemigroups and relational semigroups, on which some proofs in
this article are based. Finally, Appendix~\ref{A:three} shows the main
proofs of Section~\ref{S:lr-pow}, as already mentioned.


\section{Multisemigroups}\label{S:relational-semigroups}

Categories can be axiomatised in object-free style~\cite{MacLane98},
essentially as partial monoids with many units that satisfy a certain
locality condition.  We have previously studied their correspondence
to relational monoids~\cite{CranchDS20}, which are sets $X$ equipped
with ternary relations $R\subseteq X\times X\times X$ with many
relational units that satisfy certain relational associativity and
unit laws.  Here we use an isomorphic representation of
multioperations of type $X\times X\to \Pow X$ (the isomorphism being
that between the category $\mathbf{Rel}$ and the Kleisli category of
the powerset monad).  A multioperation on $X$, like the corresponding
ternary relation on $X$, relates pairs in $X\times X$ with a set of
elements of $X$, including the empty set. This encompasses partial and
total operations, where each such pair is related to at most one
element and to precisely one element, respectively.  Multioperations
have a long history in mathematics, see~\cite{KudryavtsevaM15} for a
discussion, references and additional material.

Dealing with categories that are not small requires classes instead of
sets in some of our examples. In the tradition of object-free
categories, we then tacitly extend the following definitions to
classes, yet do not formally distinguish between sets and
classes---the simple constructions in this article do not lead to
paradoxes.

So let $\mcomp:X\times X\to \Pow X$ be a multioperation on set $X$. We
write $D_{xy}$ in place of $x\mcomp y \neq \emptyset$ to indicate that
the composition $x\mcomp y$ is non-empty, which intuitively means
defined, and extend $\mcomp$ to an operation
$\mcomp:\Pow X\times \Pow X\to \Pow X$ defined, for all
$A,B\subseteq X$, by
\begin{equation*}
    A\mcomp B =\bigcup\{x\odot y\mid x\in A \text{ and }  y\in B\}.
\end{equation*}
We write $x\mcomp B$ instead of $\{x\}\mcomp B$, $A\mcomp x$ instead
of $A\mcomp\{x\}$ and $f(A)$ for $\{f(a)\mid a\in A\}$. Note that
$A\mcomp \emptyset=\emptyset\mcomp B=\emptyset$.

A \emph{multimagma} $(X,\mcomp)$ is simply a non-empty set $X$ with a
multioperation $\mcomp:X\times X\to \Pow X$.
\begin{itemize}
\item The multioperation $\mcomp$ is \emph{associative} if
  $x\mcomp (y\mcomp z) = (x\mcomp y)\mcomp z$.
\item It is \emph{local} if
  $u\in x\mcomp y\land D_{yz} \Rightarrow D_{uz}$ for all
  $u,x,y,z\in X$.
\item It is a \emph{partial operation} if $|x\mcomp y| \le 1$ and a
  \emph{(total) operation} if $|x\mcomp y|=1$, for all $x,y\in X$.
\item An element $e\in X$ is a \emph{left unit} in $X$ if
  $\exists x.\ x \in e\mcomp x$ and
  $\forall x,y.\ x\in e\mcomp y \Rightarrow x=y$; it is a \emph{right
    unit} in $X$ if $\exists x.\ x \in x\mcomp e$ and
  $\forall x,y.\ x\in y\mcomp e\Rightarrow x=y$. We write $E$ for the
  set of all (left or right) units in $X$. 
\end{itemize}

\begin{remark}
  Intuitively, locality states that if $x$ and $y$ as well as $y$ and
  $z$ can be composed, then every element in the composition of $x$
  and $y$ can be composed with $z$. For partial operations this
  reduces to the composition pattern of categories, binary relations,
  paths in digraphs and many other examples, as explained below.
  Locality has previously been called
  \emph{coherence}~\cite{CranchDS20}. It will become clear below why
  we are now adopting a different name.
\end{remark}

A multimagma $X$ is \emph{unital} if for every $x\in X$ there exist
$e,e'\in E$ such that $D_{ex}$ and $D_{xe'}$.  This definition of
units follows that for object-free categories. Equivalently, we may
require that there exists a set $E\subseteq X$ such that, for all
$x\in X$,
\begin{equation*}
  E\mcomp x = \{x\}\qquad\text { and } \qquad
  x\mcomp E = \{x\}.
\end{equation*}

A \emph{multisemigroup} is an associative multimagma; a
\emph{multimonoid} a unital multisemigroup.

These definitions imply that a multisemigroup $(X,\mcomp)$ is a
\emph{partial semigroup} if $\mcomp$ is a partial operation and a
\emph{semigroup} if $\mcomp$ is a total operation---and likewise for
partial monoids and monoids.

Spelling out associativity yields
\begin{equation*}
x\mcomp (y\mcomp z) =\bigcup \{x \mcomp v\mid v \in y \mcomp z\} =
\bigcup \{u \mcomp z\mid u \in x \mcomp y\}=(x\mcomp y)\mcomp z.
\end{equation*}

Multimagmas and relatives form categories in several ways. A
\emph{multimagma morphism} $f:X\to Y$ satisfies
\begin{equation*}
f(x\mcomp_X y) \subseteq f(x) \mcomp_Yf(y).
\end{equation*}
The morphism is \emph{bounded} if, in addition,
\begin{equation*}
f(x)\in u \mcomp_Y v \Rightarrow \exists y,z.\ x \in y\mcomp_X z
\land u=f(y)\land v=f(z).
\end{equation*}

Obviously, $f$ is a multimagma morphism if and only if
$x\in y \mcomp_X z\Rightarrow f(x) \in f(y) \mcomp_Y f(z)$.  This is a
natural generalisation of
$x=y\mcomp_X z \Rightarrow f(x) = f(y)\mcomp_Y f(z)$, and hence of
$f(x\mcomp_X y) = f(x)\mcomp_Y f(y)$, for total operations. For
partial operations, it implies that the right-hand side of the
inclusion must be defined whenever the left-hand side is, and, in the
bounded case, that the left-hand side is defined whenever the
right-hand side is.

A morphism $f:X\to Y$ of unital multimagmas needs to preserve units as
well: $f(e) \in E_Y$ holds for all $e\in E_X$, and $e\in E_Y$ implies
that there is an $e'\in E_X$ such that $f(e')= e$ holds if $f$ is
bounded. Morphisms of object-free categories are functors.

More generally, bounded morphisms are standard in modal logic as
functional bisimulations. The isomorphisms between categories of
relational semigroups and monoids and those of multisemigroups and
multimonoids are explained in Appendix~\ref{A:two}.

In every multimagma, every unit $e$ satisfies $e\mcomp e=\{e\}$ and
$D_{ee}$. If $e,e'\in E$, then $D_{ee'}$ if and only if $e=e'$, for if
$D_{ee'}$ holds, then $e\mcomp e'=\{x\}$ for some $x\in X$ and hence
$e=x=e'$ by the (implicational) unit axioms.  Units are therefore
idempotents which are, in a sense, orthogonal.

In every multimonoid, each element has precisely one left and one
right unit: if $e,e'\in E$ both satisfy $e\mcomp x = \{x\}=e'\mcomp x$
for some $x\in X$, then
$\emptyset \neq e\mcomp x = e\mcomp (e'\mcomp x) = (e\mcomp e')\mcomp
x$, which is only the case when $e=e'$, as explained above (the
argument for right units is similar). This
functional correspondence allows defining source and target maps
$\ell, r:X\to X$ such that $\ell(x) $ denotes the unique left unit and
$r(x)$ the unique right unit of $x$.  Then $D_{xy}$ implies
$r(x) =\ell(y)$ and the converse implication is equivalent to
locality. These properties have been proved for relational
structures~\cite{CranchDS20}, but hold in the corresponding
multialgebras via the isomorphism.

\begin{example}[Multimonoids]\label{ex:multimon}~
\begin{enumerate}
\item The object-free categories in Chapter I.1 of Mac Lane's
  book~\cite{MacLane98} and the local partial monoids are the same
  class, as already mentioned; the category of such object-free
  categories and that of local partial monoids (with both types of
  morphisms mentioned) are isomorphic.  Local partial monoids
  \emph{are} small categories, and categories if the partial monoid is
  built on a class.
\item \label{ex:monoid} Every monoid $(X,\cdot,1)$ is a one-element
  category and therefore local. The diagram
\begin{equation*}
\begin{tikzcd}
  1\arrow[out=120,in=60,loop, "a"]
\end{tikzcd}
\end{equation*}
for instance, corresponds to a partial monoid $X=\{1,a\}$ with
multiplication defined by $11=\{1\}$ and $1a=\{a\}=a1$; in addition we
can impose $aa=\{a\}$. Locality is trivial: composition is total and
$\ell(x)=1=r(x)$ for all $x\in X$.
\item \label{ex:pair-groupoid} The \emph{pair groupoid} (or
  \emph{codiscrete groupoid}) $(X\times X,\odot,\mathit{Id}_X)$ over
  any set $X$ is formed by the set of ordered pairs over $X$ with
\begin{equation*}
(w,x)\mcomp (y,z)=
\begin{cases}
  \{(w,z)\}& \text{ if } x=y,\\
\emptyset &\text{ otherwise}
\end{cases}
\qquad\quad \text{ and }\qquad E=\mathit{Id}_X,
\end{equation*}
where $\mathit{Id}_X$ is the identity relation on $X$. It is a local
partial monoid, and hence a category---in fact even a groupoid as the
name indicates. Source and target maps are given by
$\ell((x,y))= (x,x)$ and $r((x,y)) = (y,y)$. This category is
equivalent to a trivial category. $X\times X$ is nothing but the
universal relation on $X$. 

\item\label{ex:poset} The pair groupoid example generalises from
  universal relations to smaller relations $R\subseteq X\times X$.  If
  $R$ is transitive, then $\mcomp$ is well-defined and $(R,\odot)$ a
  local partial semigroup.  If $R$ is also reflexive, and hence a
  preorder, then $(R,\odot,\mathit{Id}_X)$ is a partial monoid; it is
  a groupoid if and only if $R$ is symmetric.  Finally, if $R$ is a
  partial order, we obtain the poset $(X,R)$ regarded as a poset
  category.

\item The \emph{shuffle multimonoid}
  $(\Sigma^\ast,\parallel,\{\varepsilon\})$, where $\Sigma^\ast$ is
  the free monoid over the (finite) alphabet $\Sigma$, $\varepsilon$
  is the empty word and the multioperation
  ${\parallel}:\Sigma^\ast \times\Sigma^\ast \to \Pow\Sigma^\ast$ is
  defined, for $a,b\in \Sigma$, $v,w\in \Sigma^\ast$, by
\begin{equation*}
  v\parallel \varepsilon = \{v\}=\varepsilon \parallel v\qquad\text{
    and } \qquad(av)\parallel (bw) = a (v\parallel (bw))\cup
  b ((av)\parallel w),
\end{equation*}
is not a category: $\parallel$ is a proper multioperation.  Locality
is trivial because $v\parallel w \neq \emptyset$ and
$\ell(v)=\varepsilon = r(v)$ for all $v,w\in\Sigma^\ast$.

\item \label{ex:broken-mon} The monoid in (\ref{ex:monoid}) becomes
  partial and non-local when we break composition and impose
  $aa=\emptyset$, because $\ell (a) = r (a)$ still holds.  Instead of
  a one-element category, it is now a plain digraph.

\item \label{ex:heaplets} The \emph{partial abelian monoid of
    heaplets} $(H,\mcomp,\varepsilon)$ used in separation logic is
  formed by the set $H$ of partial functions between two sets. The
  partial operation $\mcomp$ is defined as
  \begin{equation*}
      f\mcomp g =
      \begin{cases}
        f\cup g & \text{ if } \dom(f)\cap \dom(g)=\emptyset,\\
        \emptyset &\text{ otherwise}.
      \end{cases}
    \end{equation*}
    Its unit is the partial function $\varepsilon$ with empty
    domain. Locality fails because $\ell(f)=\varepsilon=r(g)$ for all
    $f,g\in H$, but $f\odot g=\emptyset$ when domains of $f$ and $g$
    overlap. This algebra of heaplets is a non-local instance of a
    \emph{generalised effect algebra}, used for modelling unsharp
    measurements in the foundations of quantum mechanics: a partially
    abelian monoid with a single unit (which in addition is
    cancellative and positive)~\cite{HedlikovaP96}.
 
\item In the multimagma $(\{x,e,e'\},\mcomp,\{e,e'\})$ with composition
  defined by
  \begin{equation*}
    \begin{array}{c|ccc}
      \mcomp & e & e' & a\\
      \hline
      e & \{e\} & \emptyset & \{a\}\\
      e'& \emptyset & \{e'\} & \{a\}\\
      a & \emptyset & \{a\} &\{a\}
    \end{array}
  \end{equation*}
  the element $a$ has left units $e$ and $e'$ and right unit
  $e'$. Associativity fails because
  $(e\mcomp e')\mcomp a\neq e\mcomp (e'\mcomp a)$. This shows that
  units need not be uniquely defined in multimagmas.\qed
   \end{enumerate}
\end{example}

Additional examples can be found in Section~\ref{S:examples}.


\section{$\ell r$-Multisemigroups}\label{S:lr-semigroups}

The correspondence between units in relational monoids and
multimonoids and source/target functions motivates alternative
axiomatisations of these structures and thus of object-free categories
that generalise the function systems of Schweizer and
Sklar~\cite{SchweizerS67}.  Such categories can be found in Chapter
XII.5 of Mac Lane's book~\cite{MacLane98}. Domain axioms for ordered
semigroupoids, with a view on allegories, have already been considered
by Kahl~\cite{Kahl08}.  We define multisemigroups with such
functions. An isomorphic alternative based on relational semigroups
has been outlined in~\cite{CranchDS20}.

An \emph{$\ell r$-multimagma} is a structure $(X,\mcomp,\ell,r)$ such that
$(X,\mcomp)$ is a multimagma and the operations $\ell,r:X\to X$
satisfy, for all $x,y\in X$,
\begin{equation*}
  D_{xy} \Rightarrow r(x) =\ell(y),\qquad
  \ell(x)\mcomp x = \{x\},\qquad
  x\mcomp r(x) = \{x\}.
\end{equation*}
Recall that $D_{xy}\Leftrightarrow x\mcomp y \neq
\emptyset$.
Henceforth we often write $xy$ instead of $x\mcomp y$ and $AB$ instead
of $A\mcomp B$.

An \emph{$\ell r$-multisemigroup} is an associative
$\ell r$-multimagma.  An $\ell r$-multimagma is \emph{$\ell r$-local}
whenever $r(x) = \ell (y) \Rightarrow D_{xy}$ and therefore
$D_{xy} \Leftrightarrow r(x) = \ell (y)$.  Locality in the sense of
multimonoids and $\ell r$-locality coincide in
$\ell r$-multisemigroups; henceforth we simply speak about locality of
$\ell r$-multisemigroups.

\emph{Duality} for $\ell r$-multimagmas arises by interchanging $\ell$
and $r$ as well as the arguments of $\mcomp$.  The classes of
$\ell r$-multimagmas and semigroups, hence also
$\ell r$-multisemigroups, are closed under this transformation;
locality and partiality are self-dual. Hence the dual of any property
that holds in any of these classes, obtained by making these
replacements, holds as well. This generalises opposition in
categories.

\begin{lemma}\label{P:lr-magma-props}
  In any $\ell r$-multimagma,
  \begin{enumerate}
  \item the \emph{compatibility laws} $\ell \circ r= r$, $r\circ \ell =
    \ell$ and \emph{retraction laws} $\ell\circ \ell =
    \ell$, $r\circ r =r$ hold, 
  \item the \emph{idempotency law} $\ell(x)\ell(x) = \{\ell(x)\}$ holds,
  \item the \emph{commutativity law} $r(x)\ell(y) = \ell(y)r(x)$ holds,
  \item the \emph{export laws} $\ell(\ell(x)y) = \ell(x)\ell(y)$ and
    $r(xr(y)) = r(x)r(y)$ hold, 
\item the \emph{weak twisted laws} $\ell(xy) x \subseteq x \ell(y)$ and $x r(yx) \subseteq r(y)x$ hold.
  \end{enumerate}
\end{lemma}

All proofs have been checked with Isabelle.

\begin{remark}
  We compare our axioms and derived laws with the axioms of Schweizer
  and Sklar's function systems~\cite{SchweizerS67} (see
  Appendix~\ref{A:one}) for a list).  The associativity axiom of
  multisemigroups generalises the associativity axiom of function
  systems.  The compatibility laws are their Axioms (\ref{eq:2a}), the
  \emph{absorption axioms} $\ell(x)x=\{x\}=xr(x)$ their Axioms
  (\ref{eq:2b}); the commutativity law is their Axiom (\ref{eq:3b}).
  The export laws are Axioms (\ref{eq:D3}) and its opposite (R3) of
  modal semigroups~\cite{DesharnaisJS09}.  The relationship between
  $\ell r$-algebras, function systems and modal semigroups is
  summarised in  Remark~\ref{Remark:Relationships} below.
\end{remark}

The compatibility laws imply that
$\ell (x) = x\Leftrightarrow r (x) = x$ and further that
\begin{equation*}
X_\ell = \{x\mid \ell (x) = x\} = \{x\mid r (x)=x\} = X_r.
\end{equation*}
Moreover, by the retraction laws, $X_\ell = \ell (X)$ and
$X_r = r (X)$.

Lemma~\ref{P:lr-magma-props} implies additional laws, including
\begin{itemize}
\item the commutativity laws $\ell(x)\ell(y) = \ell(y)\ell(x)$ and
  $r(x) r(y)=r(y) r(x)$,
\item the idempotency law 
  $r(x) r(x)= \{r(x)\}$,
\item the orthogonality laws
  $D_{\ell(x)\ell(y)} \Leftrightarrow \ell(x)=\ell(y)$ and
  $D_{r(x)r(y)} \Leftrightarrow r(x)=r(y)$.
\end{itemize}
Finally, every $\ell r$-multimagma is unital with $E=X_\ell=X_r$, and
we often write $E$ instead of $X_\ell$ or~$X_r$.

\begin{lemma}\label{P:lr-semigroup-props}
  In any  $\ell r$-multisemigroup, 
  \begin{enumerate}
  \item the \emph{weak locality laws} $\ell (xy) \subseteq \ell(x\ell(y))$
and $r (xy) \subseteq  r(r(x)y)$ hold,
\item the \emph{conditional locality laws}
  $D_{xy}\Rightarrow \ell (xy) = \ell(x\ell(y))$
  and $D_{xy}\Rightarrow r (xy) = r(r(x)y)$ hold,
\item the laws $\ell (xy)\subseteq \{\ell(x)\}$ and
 $r(xy)\subseteq \{r(y)\}$ hold, 
\item the conditional laws
  $D_{xy} \Rightarrow \ell (xy) = \{\ell(x)\}$ and
  $D_{xy} \Rightarrow r(xy) = \{r(y)\}$ hold,
\item the \emph{conditional twisted laws} 
  $D_{xy}\Rightarrow \ell(xy)x = x\ell(y)$ and $D_{xy}\Rightarrow y
  r(xy) = r(x)y$ hold.
  \end{enumerate}
\end{lemma}

The proofs have again been checked with Isabelle.

\begin{remark}\label{Remark:Relationships}
  The locality laws generalise Axioms (\ref{eq:3a}) for function
  systems~\cite{SchweizerS67}; the twisted laws generalise Axiom
  (\ref{eq:3c}) and law (\ref{eq:D3c}), in this order.  Function
  semigroups without (\ref{eq:3c}) and law (\ref{eq:D3c}) axiomatise
  modal semigroups~\cite{DesharnaisJS09}, which relate to semigroups
  of binary relations. Adding Axiom (\ref{eq:3c}) to both classes
  axiomatises domain and codomain of functions. Law (\ref{eq:D3c}) is
  relevant for systems of functions with so-called
  subinverses~\cite{SchweizerS67}. These are related to inverse
  semigroups and are irrelevant for this article. In sum,
  $\ell r$-multisemigroups generalise function systems and modal
  semigroups beyond totality.  A discussion of related semigroups and
  their applications can be found in~\cite{DesharnaisJS09}.
\end{remark}

\begin{lemma}\label{P:coherent-lr-semigroup-props}
In any local $\ell r$-multisemigroup,
\begin{enumerate}
\item the \emph{equational locality laws} $\ell (xy) =\ell(x\ell(y))$
and $r (xy) = r(r(x)y)$ hold,
\item the \emph{twisted laws} $\ell(xy)x = x\ell(y)$ and
  $y r(xy) = r(x)y$ hold.
\end{enumerate}
\end{lemma}

Once again, all proofs have been checked with Isabelle.  Locality is
in fact an equational property.

\begin{proposition}\label{P:local-coherent}
  An $\ell r$-multisemigroup is $\ell r$-local if and only if
  $\ell(x\ell(y)) \subseteq \ell (xy)$ and
  $r(r(x)y) \subseteq r (xy)$.
  \end{proposition}
  \begin{proof}
    We have checked with Isabelle that the equational locality laws
    imply $\ell r$-locality in any $\ell r$-multimagma. Equality in
    $\ell r$-multisemigroups follows from
    Lemma~\ref{P:coherent-lr-semigroup-props}.
  \end{proof} 

\begin{remark}
  Locality and weak locality have already been studied in the context
  of predomain, precodomain, domain and codomain operations for
  semirings~\cite{DesharnaisMS06}. In this context, predomain and
  precodomain operations satisfy weak locality axioms, but not the
  strong ones. Relative to $\ell r$-multisemigroups, these variants of
  domain and codomain axioms are at the powerset level.  Our results
  in Sections~\ref{S:lr-pow} to~\ref{S:correspondences} explain how
  they are caused by locality in $\ell r$-multisemigroups.
  \end{remark}

 \begin{remark}
   It seems natural to ask whether the axiom
   $D_{xy}\Rightarrow r(x)=\ell(y)$ in the definition of
   $\ell r$-magmas could be replaced, like locality, by equational
   axioms. Experiments with Isabelle show that adding the equational
   properties derived in
   Lemmas~\ref{P:lr-magma-props}-\ref{P:coherent-lr-semigroup-props}
   does not suffice. We leave this question open.
  \end{remark}

  The final lemma on $\ell r$-multisemigroups yields a more
  fine-grained view on definedness conditions and $\ell r$-locality.
\begin{lemma}\label{P:local-alt}
In any $\ell r$-multimagma, 
\begin{enumerate}
\item $r(x) = \ell(y) \Leftrightarrow D_{r(x)\ell(y) }$ and $D_{xy}  \Rightarrow D_{r(x)\ell(y)}$.
\item $D_{r(x)\ell(y)} \Rightarrow D_{xy}$ whenever the
  $\ell r$-multimagma is a local $\ell r$-multisemigroup.
\end{enumerate}
\end{lemma}
Once again, the proofs have been obtained by Isabelle.  The
correspondence between multimonoids and $\ell r$-multisemigroups can
now be summarised as follows.
\begin{proposition}\label{P:lr-semigroup-rel-monoid}
  Every multimonoid $(X,\mcomp,E)$  is an $\ell r$-multisemigroup
  $(X,\mcomp,\ell, r)$ in which $\ell (x)$ and $r (x)$ indicate the
  unique left and right unit of any $x\in X$.  Conversely, every
  $\ell r$-multisemigroup $(X,\mcomp,\ell,r)$ is a multimonoid
  $(X,\mcomp,E)$ with $E=X_\ell=X_r$.
\end{proposition}
The result trivially carries over to local structures and extends
to isomorphisms between categories of $\ell r$-multisemigroups and
relational monoids with suitable morphisms.  A morphism $f$ of
$\ell r$-multimagmas $(X,\mcomp_X,\ell_X,r_X)$ and
$(Y,\mcomp_Y,\ell_Y,r_Y)$ is of course a multimagma morphism that satisfies
$f\circ\ell_X = \ell_Y\circ f$ and $f\circ r_X=r_Y\circ f$. 

\begin{example}[$\ell r$-Multisemigroups]\label{ex:lr-mgs}
  All structures in Example~\ref{ex:multimon} are
  $\ell r$-multisemigroups by
  Proposition~\ref{P:lr-semigroup-rel-monoid}. We consider some of
  them in more detail.
\begin{enumerate}
\item Local partial $\ell r$-semigroups and the object-free categories
  in Chapter XII.5 of Mac Lane's book~\cite{MacLane98} are the same
  class, as already mentioned; the categories of such object-free
  categories and local partial $\ell r$-semigroups (with both kinds of
  morphisms) are isomorphic. Hence local partial $\ell r$-semigroups
  \emph{are} small categories. We briefly recall the relationship to
  more standard definitions of categories.

  A (small) \emph{category} consists of a set $O$ of objects and a set
  $M$ of morphisms with maps $s,t:M\to O$ associating a source and
  target object with each morphism, an operation $\id:O\to M$
  associating an identity arrow with each object, and a partial
  operation of composition $;$ of morphisms such that $f; g$ is
  defined whenever $t(f) = s(g)$---where $f;g=g\circ f$. The following
  axioms hold for all $X\in O$ and $f,g,h\in M$ (for Kleene equality,
  that is, whenever compositions are defined):
\begin{gather*}
  s(\id(X)) = X,\qquad t(\id(X)) = X,\qquad
s (f; g) = s(f),\qquad  t (f;g) = t(g),\\
f ; (g;h) = (f; g); h,\qquad
\id (s (f)) ; f = f,\qquad f ; \id (t (f)) = f.
\end{gather*}
In a category, $\ell = \id\circ s$ and $r = \id\circ t$ leads to
$\ell r$-multisemigroups. Conversely, the elements in $E$ of a local
partial $\ell r$-semigroup $X$ serve as objects of a category. 
\item The pair groupoid is a local partial $\ell r$-semigroup. 
\item \label{ex:broken-mon-bis} In the broken monoid from
  Example~\ref{ex:multimon}(\ref{ex:broken-mon}),
  $ \ell (aa) =\ell(\emptyset)=\emptyset \subset \{1\} = \ell (a1) =
  \ell(a\ell(a))$, hence locality of $\ell$ fails; that of $r$ fails
  by duality.

\item \label{ex:heaplets2} In the partial abelian $\ell r$-semigroup of heaplets from
  Example~\ref{ex:multimon}(\ref{ex:heaplets}), locality fails, too:
  if the domains of $f,g\in H$ overlap, then
  $\ell(fg)=\ell(\emptyset) = \emptyset \subset \{\varepsilon\} =
  \ell(f)=\ell(f\varepsilon) = \ell(f\ell(g))$
  and locality of $r$ fails by duality.
  
\item \label{ex:matrices} A concrete example of a category as a local
  partial $\ell r$-semigroup are Elgot's \emph{matrix theories}
  \cite{Elgot76}.  The set
  $\smash[b]{\mathbb{M}S= \bigcup_{ n, m\ge 0} S^{ n\times m}}$ of
  matrices over any semiring $S$ forms a partial monoid with matrix
  multiplication as composition.  $\mathbb{M}S$ is a local partial
  $\ell r$-semigroup with $\ell$ and $r$ defined, for any
  $M\in S^{ n\times m}$, by $\ell( M)= I_n$ and $r( M)= I_m$: the
  identity matrices of the appropriate dimensions. This category is
  isomorphic to the standard category of matrices, which has natural
  numbers as objects and $n\times m$-matrices as elements of the
  hom-set $[n,m]$.\qed
\end{enumerate}
\end{example}


\section{Convolution Quantales}\label{S:quantales}

We have already extended the multioperation
$\mcomp:X\times X\to \Pow X$ to type $\Pow X\times \Pow X\to \Pow X$
defining
$A \mcomp B = \bigcup\{x\mcomp y\mid x\in A \text{ and } y\in B\}$ and
the maps $\ell,r:X\to X$ to type $\Pow X\to \Pow X$ by taking
images. We wish to explore the algebraic structure of such powerset
liftings over $\ell r$-multimagmas and their relatives. It is known
that powerset liftings of relational monoids yield unital
quantales~\cite{Rosenthal97,DongolHS21}. Yet the precise lifting of
source and target operations remains to be explored. This requires
some preparation.

A \emph{quantale}~\cite{Rosenthal90} $(Q,\le,\cdot,1)$ is a complete
lattice $(Q,\le)$ equipped with a monoidal composition $\cdot$ with
unit $1$ that preserves all sups in its first and second argument.

We write $\Sup$ for the sup and $\Inf$ for the inf operator, and
$\sup$ and $\inf$ for their binary variants.  We also write
$\bot=\Inf Q=\Sup\emptyset$ for the least and
$\top = \Sup Q=\Inf\emptyset$ for the greatest element of $Q$.  In the
literature, quantales are often defined without $1$ (those with $1$
are then called \emph{unital}), but we have no use for these. We write
$Q_1=\{\alpha\in Q\mid \alpha\le 1\}$ for the set of
\emph{subidentities} of quantale $Q$.

A quantale is \emph{boolean} if its lattice reduct is a complete
boolean algebra---a complete lattice and a boolean algebra.

We write $-$ for boolean complementation. In a boolean quantale, $Q_1$
forms a complete boolean subalgebra with complementation
$\lambda x.\ x'= \lambda x.\ 1 - x$, and in which composition
coincides with meet~\cite{FahrenbergJSZ20}.

Our examples in Section~\ref{S:examples} require weaker notions of
quantale. A \emph{prequantale} is a quantale in which the
associativity law is absent~\cite{Rosenthal90}.

\begin{proposition}\label{P:lr-pow}
  Let $(X,\mcomp,\ell,r)$ be an $\ell r$-multisemigroup. Then
  $(\Pow X,\subseteq,\odot,E)$ forms a boolean quantale in
  which the complete boolean algebra is atomic.
\end{proposition}
\begin{proof}
  If $(X,\mcomp,\ell,r)$ is an $\ell r$-multisemigroup, then $(X,\mcomp,E)$ is a
  multimonoid by Proposition~\ref{P:lr-semigroup-rel-monoid}, hence
  isomorphic to a relational monoid 
  and its powerset algebra a quantale~\cite{Rosenthal97} (see
  also~\cite{DongolHS21}).  The complete lattice on $\Pow X$ is
  trivially boolean atomic.
\end{proof}

\begin{remark}
  If $X$ is an $\ell r$-multimagma instead of a
  $\ell r$-multisemigroup, then $\Pow X$ forms a prequantale instead
  of a quantale~\cite{DongolHS21,CranchDS20a}.  This weaker result is
  needed in Section~\ref{S:lr-pow}.
\end{remark}

\begin{remark}
Dualities between $n+1$-ary relational structures and boolean algebras
with $n$-ary (modal) operators have been studied by J\'onsson and
Tarski~\cite{JonssonT51}; correspondences between relational
associativity laws and those at powerset level are well known from
substructural logics such as the Lambek calculus~\cite{Lambek58}. Here
we study ternary relations---as binary multioperations---and the
quantalic composition as a binary (modal) operation.
\end{remark}

\begin{example}[Powerset Quantales over
  $\ell r$-Semigroups]\label{ex:pow-quantales}
  While the lifting works for arbitrary $\ell r$-multisemigroups, we
  restrict our attention to categories.
\begin{enumerate}
\item Let $C=(O,M)$ be a category. Then $(\Pow C,\subseteq,\mcomp,1)$,
  with the operations from Section~\ref{S:quantales} and
  $1=\{\id_X\mid X\in O\}$, forms a boolean quantale, in fact an
  atomic boolean one.  This holds by Proposition~\ref{P:lr-pow}
  because categories are $\ell r$-semigroups
  (Example~\ref{ex:multimon}). Note the difference between this
  construction and that of the powerset functor: we take powersets of
  morphisms, not of objects.
\item An interesting instance lifts the pair groupoid on set $X$ to
  the quantale of binary relations over $X$.  The quantalic
  composition is relational composition, the monoidal unit the
  identity relation, set union is sup and set inclusion the partial
  order. Relations can be seen as possibly infinite-di\-men\-sional
  boolean-valued square matrices in which the quantalic composition is
  matrix multiplication (cf.\ Example
  \ref{ex:lr-mgs}(\ref{ex:matrices})). \qed
\end{enumerate}
\end{example}
The fact that groupoids can be lifted to algebras of binary relations
with an additional operation of converse,
$R^\smallsmile=\{(y,x)\mid (x,y)\in R\}$, was known to J\'onsson and
Tarski~\cite{JonssonT52}.  This applies in particular to groups as
single-object groupoids. Further examples of powerset liftings of
categories and other $\ell r$-multisemigroups can be found in
Section~\ref{S:examples}. 

The powerset lifting to $\Pow X$ is a lifting to the function space
$2^X$, where $2$ is the quantale of booleans. It generalises to
liftings to function spaces $Q^X$ for arbitrary quantales
$Q$~\cite{DongolHS21}. We present a multioperational version.

Let $(X,\mcomp,\ell,r)$ be an $\ell r$-multimagma and
$(Q,\le,\cdot,1)$ a quantale. For functions $f,g:X\to Q$, we define
the \emph{convolution} $\ast: Q^X\times Q^X\to Q^X$ as
\begin{equation*}
  (f\ast g) (x) = \Sup_{x\in y\mcomp z} f (y) \cdot g (z)=\Sup\{f (y)
  \cdot g (z) \mid x\in y\mcomp z\}.
\end{equation*}
For any predicate $P$ we define
\begin{equation*}
  [P]=
  \begin{cases}
    1 & \text{ if } P,\\
\bot & \text{ otherwise},
  \end{cases}
\end{equation*}
and then the function $\id_{E}:X\to Q$ as
\begin{equation*}
  \id_E (x) = [x\in E].
\end{equation*}

In addition, we extend sups and $\le$ pointwise from $Q$ to
$Q^X$. This leads to the following generalisation of
Proposition~\ref{P:lr-pow}.
\begin{theorem}\label{P:lr-conv}
Let $X$ be an $\ell r$-multisemigroup and $Q$ a quantale. Then
$(Q^X,\le,\ast,\id_E)$ is a quantale.
\end{theorem}
\begin{proof}
  If $(X,\mcomp,\ell,r)$ is an $\ell r$-multisemigroup and $Q$ a quantale,
  then $(X,\mcomp,X_E)$ is a multimonoid by
  Proposition~\ref{P:lr-semigroup-rel-monoid} and a relational monoid
  up-to isomorphism. Hence $Q^X$ is a quantale~\cite{DongolHS21}
  (shown for slightly different, but equivalent relational monoid
  axioms).
\end{proof}
In addition, $Q^X$ is distributive (more precisely, its underlying
complete lattice) if $Q$ is, and boolean if $Q$ is~\cite{DongolHS21}.

We call $Q^X$ from Theorem~\ref{P:lr-conv} the
\emph{convolution algebra} or \emph{convolution quantale} of $X$ and
$Q$.

\begin{remark}
  If $X$ is an $\ell r$-multimagma and $Q$ a prequantale, then the
  convolution algebra $Q^X$is a
  prequantale~\cite{DongolHS21,CranchDS20a}.  This result is needed in
  Section~\ref{S:conv-alg}.
\end{remark}

The following laws help calculating with convolutions. First of all,
using the function $\delta_x(y)=[x=y]$, we can write
$f(x) = \Sup_{y\in Y} f(y)\cdot \delta_{y}(x)$ for any $f\in Q^X$ and
more generally
\begin{equation*}
  f = \Sup_{x\in X} f(x)\cdot \delta_x,
\end{equation*}
viewing $f(x)\cdot \delta_y$ and even $\alpha\cdot f$, for any
$\alpha\in Q$ and $f:X\to Q$, as a scalar product in a $Q$-module on
$Q^X$, as usual in algebra.  This allows us to rewrite convolution as
\begin{equation*}
  (f \ast g)(x) = \Sup_{y,z\in X} f(y)\cdot g(z) \cdot [x\in y\odot z]
\end{equation*}
or even as
\begin{equation*}
  f \ast g = \Sup_{x,y,z\in X} f(y)\cdot g(z) \cdot [x\in
  y\odot z]\cdot \delta_x, 
\end{equation*}
sups as
\begin{equation*}
  \Sup F= \Sup_{x\in X} \Sup\{f(x)\mid f\in F\}\cdot \delta_x  \qquad\text{ and  }\qquad (f+ g) = \Sup_{x\in X} (f(x)+g(x))\cdot
  \delta_x,
\end{equation*}
and finally the identity function as $\id_E = \Sup_{e\in E}\delta_e$.

\begin{example}[Convolution Quantales over $\ell r$-Semigroups]\label{ex:category-quantale}
Once again we restrict our attention to categories. 
\begin{enumerate}
\item A \emph{category algebra} is the convolution algebra of a small
  category with values in a commutative ring with unity.  This
  generalises the well known group algebras, in particular group
  rings. Similarly, Theorem~\ref{P:lr-conv} constructs \emph{category
    quantales}, evaluating small categories in quantales, which yields
  quantales as convolution algebras.
\item An instance are $Q$-fuzzy relations~\cite{Goguen67}, which are
  binary relations taking values in a quantale $Q$. The associated
  quantales are convolution quantales over pair groupoids. If $Q$ is
  the Lawvere quantale described in Example~\ref{ex:value-quantales}
  below, this yields t-norms.  $Q$-valued relations correspond to
  possibly infinite dimensional $Q$-valued square matrices.  If the
  base set is finite, we recover the matrix theories $\mathbb{M}Q$ of
  Example \ref{ex:lr-mgs}(\ref{ex:matrices})d.  Heisenberg's original
  formalisation of quantum mechanics~\cite{Heisenberg25} used a
  similar convolution algebra over the pair groupoid~\cite{Connes95},
  yet with values in the field of complex numbers. \qed
\end{enumerate}
\end{example}

Many additional examples of  relational monoids, partial monoids and
convolution algebras are discussed
in~\cite{DongolHS16,DongolHS21,CranchDS20a}. Further examples, with a
view of lifting categories and non-local $\ell r$-semigroups, can be
found in Section~\ref{S:examples}. 

Finally, we list some well known candidates for value-quantales.

\begin{example}[Quantales]\label{ex:value-quantales}~
\begin{enumerate}
\item\label{ex:bool-quantale} We have already mentioned the quantale
  of booleans $(2,\le,\inf,1)$, which has carrier set $\{0,1\}$ and
  $\inf$, in fact $\min$, as composition.

\item \label{ex:lawvere-quantale} The \emph{Lawvere quantale}
  $(\mathbb{R}_+^\infty,\ge,+,0)$ has $\Inf$ as supremum, $+$ as
  quantalic composition, extended by $x+\infty=\infty=\infty +x$, and
  $0$ as its unit. It is important for defining generalised metric
  spaces and t-norms.

\item The unit interval $([0,1],\le,\cdot,1)$ forms a quantale with
  $\Sup$ as supremum. It is isomorphic to the Lawvere quantale via the
  function $(\lambda x \cdot -\ln x)$ as its inverse, and important in
  probability applications.

\item The structures $([0,1],\le,\min,1)$ and $([0,1],\ge,\max,0)$ and
  their variants with $[0,1]$ replaced by $\mathbb{R}_+^\infty$ (and
  unit $\infty$ for the first) form other quantales over the unit
  interval, which are at the basis of \emph{max-plus algebra}
  \cite{HeidergottOW06} and similar structures.\qed
\end{enumerate}
\end{example}
All these examples are distributive, but not boolean quantales, and
they justify our approach to domain in non-boolean cases below.


\section{Modal Quantales}\label{S:modal-quantales}

The results of Section~\ref{S:quantales} do not lift the source and
target structure of $\ell r$-multisemigroups to $Q^X$ faithfully: all
relational units are mapped to the unit in $Q$ by $\id_E$
rather bluntly. A more fine-grained approach, in which different
elements of $E$ are mapped to different elements in the powerset
quantale or different values or weights in $Q$, is possible.

\begin{example}[Relational domain and codomain]\label{ex:rel-dom-cod}
  In the relation quantale, the standard relational domain and
  codomain operations are
  \begin{equation*}
    \dom(R) = \{(a,a)\mid \exists b.\ (a,b)\in R\}\qquad\text{ and } \qquad
    \cod(R) = \{(b,b)\mid \exists a.\ (a,b) \in R\}.
  \end{equation*}
  Thus $\dom (R) = \{\ell (x) \mid x\in R\} = \ell (R)$ and
  $\cod (R)= \{r (x) \mid x\in R\}= r (R)$.\qed
\end{example}

We can capture this more abstractly.  Formally, a \emph{domain
  quantale}~\cite{FahrenbergJSZ20} is a quantale $(Q,\le,\cdot,1)$
equipped with a domain operation $\dom:Q\to Q$ that satisfies, for all
$\alpha,\beta\in Q$,
\begin{align*}
 \alpha &\le \dom (\alpha)\cdot \alpha,\\
\dom\, (\alpha\cdot \dom (\beta)) &= \dom (\alpha\cdot \beta),\\
 \dom (\alpha) &\le 1,\\
\dom (\bot) &= \bot,\\
\dom (\alpha \sup \beta) &= \dom (\alpha)\sup \dom (\beta).
\end{align*}
We refer to the domain axioms as \emph{absorption}, \emph{locality},
\emph{subidentity}, \emph{strictness} and \emph{(binary) sup
  preservation}, respectively. Absorption can be strengthened to an
identity: $\dom (\alpha) \alpha = \alpha$.

The domain axioms are precisely those of domain
semirings~\cite{DesharnaisS11}; domain quantales are thus quantales
that are also domain semirings with addition as binary sup.
Properties of domain semirings are therefore inherited, for instance
\begin{itemize}
\item the export law $\dom (\dom (\alpha) \beta) = \dom (\alpha) \dom (\beta)$;
\item order preservation $\alpha\le \beta\Rightarrow \dom (\alpha) \le \dom (\beta)$;
\item the weak twisted law: $\alpha\dom (\beta) \le \dom(\alpha\beta) \alpha$;
\item least left absorption (lla): $\dom (\alpha) \le \rho \Leftrightarrow
  \alpha\le \rho\alpha$; and
\item the adjunction $\dom (\alpha) \le \rho \Leftrightarrow \alpha \le \rho\top$.
\end{itemize}
In the last two laws,
$\rho\in Q_\dom=\{\alpha\mid \dom(\alpha)=\alpha\}$ (they need not
hold for all $\rho\in Q_1\supseteq Q_\dom$).

Domain axioms for arbitrary sups are
unnecessary~\cite{FahrenbergJSZ20}. In every domain quantale, $\dom$
preserves arbitrary sups, $\dom \left(\Sup A\right) = \Sup \dom (A)$,
hence in particular $\dom (\top) = 1$, and distributes weakly with
infs, $\dom \left(\Inf A\right) \le \Inf \dom (A)$. Domain elements
also left-distribute over non-empty infs,
$\dom (\alpha) \cdot \Inf A = \Inf (\dom (\alpha) A)$ for all $A\neq \emptyset$.

Much of the structure of the domain algebra induced by $\dom$ is
inherited from domain semirings as well.  It holds that
$Q_\dom = \dom(Q)$. The \emph{domain algebra} $(Q_\dom,\le,\cdot,1)$
is therefore a subquantale of $Q$ that forms a bounded distributive
lattice with $\cdot$ as binary inf. It contains the largest boolean
subalgebra of $Q$ bounded by $\bot$ and $1$~\cite{DesharnaisS11}.  The
elements of $Q_\dom$ are called \emph{domain elements} of $Q$. Yet, in
the quantalic case, the lattice $Q_\dom$ is
complete~\cite{FahrenbergJSZ20}: $\dom (\Sup\dom (A)) = \Sup \dom (A)$
follows from sup-preservation of $\dom$. All sups of domain elements
are therefore again domain elements, but sups and infs in $Q_\dom$
need not coincide with those in $Q$.

Quantales are closed under opposition: interchanging the order of
composition in quantale $Q$ yields a quantale $Q^\op$; properties of
quantales translate under this duality. The opposite of the domain
operation on a domain quantale is of course a codomain operation.

A \emph{codomain quantale} is the opposite of a domain quantale, just
like a codomain semiring is the opposite of a domain semiring.
Codomain quantales can be axiomatised using a codomain operation
$\cod:Q\to Q$ that satisfies the dual domain axioms, making
$(Q^\op,\cod)$ a domain quantale.

A \emph{modal quantale} is a domain and codomain quantale
$(Q,\le,\cdot,1,\dom,\cod)$ that satisfies the following
\emph{compatibility} axioms, which make the domain and codomain
algebras $Q_\dom$ and $Q_\cod$ coincide:
\begin{equation*}
  \dom \circ \cod = \cod \qquad\text{ and }\qquad \cod\circ \dom = \dom.
\end{equation*}

\begin{remark}
  In modal semirings, $\dom$ and $\cod$ are usually modelled
  indirectly through their boolean complements in the subalgebras of
  subidentities, that is, by antidomain and antirange operations.
  This allows expressing boolean complementation in $Q_\dom$, which
  then becomes the largest complete boolean subalgebra of $Q$ bounded
  by $\bot$ and $1$~\cite{DesharnaisS11}. As we intend to lift from
  $\ell r$-multisemigroups, where complements of source and target
  operations may not exist, we do not follow this approach.
\end{remark}

In a boolean quantale $Q$, the subalgebra $Q_1$ of subidentities forms
a complete boolean algebra with quantalic composition as binary inf,
hence in particular $Q_\dom =Q_1$. One can then axiomatise domain by
the adjunction
$\dom (\alpha) \le \rho \Leftrightarrow \alpha\le \rho\top$, for all
$\rho\in Q_1$, and weak locality
$\dom (\alpha\beta) \le \dom (\alpha\dom
(\beta))$~\cite{FahrenbergJSZ20}.
Dual results hold for codomain. Using the adjunction alone yields only
\emph{predomain} and \emph{precodomain}
operations~\cite{DesharnaisMS06}.  Finally, in a boolean quantale,
antidomain and anticodomain operations can be defined as
$\ad = (\lambda x.\ x')\circ \dom$ and its dual. The axioms of
antidomain and anticodomain semirings~\cite{DesharnaisS11} can then be
derived~\cite{FahrenbergJSZ20}.

Some of the $\ell r$-magmas, $\ell r$-semigroups and $\ell r$-monoids
in the examples in Section~\ref{S:examples} fail to yield
associativity or locality laws when lifted. This requires a more
fine-grained view on modal quantales.

\begin{itemize}
\item A \emph{modal prequantale} is a prequantale in which the locality
axioms for $\dom$ and $\cod$ are replaced by the export axioms
\begin{equation*}
  \dom(\dom(\alpha)\beta)) = \dom(\alpha)\dom(\beta)\qquad\text{ and }\qquad
  \cod(\alpha\cod(\beta)) = \cod(\alpha)\cod(\beta).
\end{equation*}
The algebra $Q_\dom=\dom(Q)=\cod(Q)=Q_\cod$ still forms a complete
distributive lattice in this case. Yet neither the weak locality laws
$\dom(\alpha\beta) \le \dom(\alpha\dom(\beta)$ or
$\dom(\alpha\beta) \ge \dom(\alpha \dom(\beta)$ nor their opposites
for $\cod$ are derivable in this setting.

\item A \emph{weakly local modal quantale} is a modal quantale in
  which the locality axioms for $\dom$ and $\cod$ have once again been
  replaced by the export axioms above.

In the presence of associativity, the \emph{weak locality} laws
\begin{equation*}
\dom(\alpha\beta) \le \dom(\alpha \dom(\beta)) \qquad\text{ and }\qquad 
\cod(\alpha\beta) \le \cod(\cod(\alpha)\beta) 
\end{equation*}
are now derivable, but not the opposite inequalities; locality
therefore does not hold.
\end{itemize}

As mentioned before, this foliation of definitions is reflected by
lifting and modal correspondence properties in
Sections~\ref{S:lr-pow}, \ref{S:conv-alg} and \ref{S:correspondences}
and justified by mathematically meaningful examples in
Section~\ref{S:examples}.

The definitions and properties mentioned have been verified with
Isabelle, they also transfer to dioids, but have not yet been
developed in greater detail for these.


\section{Modal Powerset  Quantales}\label{S:lr-pow}

We show how $\ell r$-multisemigroups can be lifted to powerset
algebras, using the lifted operations $\ell,r:\Pow X\to \Pow X$
defined, for all $A\subseteq X$, by
$\ell (A) = \{\ell (x) \mid x\in A\}$ and
$r (A) = \{r (x) \mid x\in A\}$.

\begin{remark}
  To relate this with the lifting of the ternary relation $R$ to a
  binary modality $\mcomp$ on the powerset quantale in the sense of
  J\'onsson and Tarski, consider
  $\mathcal{R}(\ell) = \{(x,y)\mid y = \ell (x)\}$, the graph of
  $\ell$, which is a functional binary relation, and similarly the
  graph $\mathcal{R}(r)$ of $r$. Then
\begin{equation*}
\ell(A) = \{y\mid (x,y)\in \mathcal{R}(\ell) \text{ for some } x\in
A\}\qquad\text{ and }\qquad
r(A) = \{y\mid (x,y)\in \mathcal{R}(r) \text{ for some } x\in A\},
\end{equation*}
which shows that $\ell$ and $r$ are modal diamond operators on
$\Pow X$. The following lemmas show that they indeed preserve binary
sups and are strict, and hence are operators on boolean algebras.
\end{remark}

Our main aim is to verify the domain quantale axioms and then use
properties of domain and codomain to infer that local
$\ell r$-multisemigroups can be lifted to boolean modal quantales at
powerset level. We develop this theorem step-by-step, starting from
$\ell r$-multimagmas, to clarify correspondences.

\begin{lemma}\label{P:lr-magma-lift}
Let $X$ be an $\ell r$-multimagma. For $A,B\subseteq
X$ and $\mathcal{A}\subseteq \Pow X$, 
\begin{enumerate}
\item  the compatibility laws $\ell (r(A)) = r (A)$ and $r (\ell (A)) =
  \ell (A)$ hold,
\item the absorption laws $\ell (A)\cdot A = A$ and $A\cdot r (A)= A$
  hold,
\item the sup-preservation laws $\ell \left(\bigcup \mathcal{A}\right) =
   \bigcup\{\ell (A)\mid A \in
  \mathcal{A}\}$ and $r \left(\bigcup \mathcal{A}\right)= \bigcup\{r (A)\mid A
  \in \mathcal{A}\}$ hold,
\item the binary sup-preservation laws  $\ell (A\cup B)= \ell (A) \cup \ell (B)$,
  $r (A\cup B)= r (A) \cup r (B)$ and the zero laws
  $\ell (\emptyset) = \emptyset = r (\emptyset)$ hold,
\item the commutativity laws $f (A) g (B) = g (B) f (A)$ hold for $f,g\in\{\ell,r\}$,
\item the subidentity laws $\ell (A) \subseteq X_\ell$ and $r (A)
  \subseteq X_r$ hold,
\item the export laws $\ell (\ell (A)\cdot  B) = \ell (A) \ell (B)$
  and $r (A\cdot  r(B)) = r (A)r (B)$ hold.
\end{enumerate}
\end{lemma}

The proof has been verified with Isabelle and is subsumed by that of
Theorem~\ref{P:lr-lifting}. Yet because powerset liftings are
important and the proof may be instructive, we present details in
Appendix~\ref{A:three}.  Lemma~\ref{P:lr-magma-lift} shows that the
domain axioms, except locality, can already be lifted from
$\ell r$-multimagmas.  The identities in (4) are subsumed by those in
(3) in the powerset quantale, but they are needed for lifting to modal
semirings; see Section~\ref{S:conv-alg}. That is why we are listing
them.

Lifting weak variants of locality for $\ell$ and $r$ requires
$\ell r$-multisemigroups; lifting locality, in addition, requires
locality.

\begin{lemma}\label{P:lr-semigroup-lift}
Let $X$ be an $\ell r$-multisemigroup and $A,B\subseteq
X$. Then
\begin{equation*}
  \ell (AB) \subseteq \ell (A \ell (B))\qquad\text{
    and }\qquad r (AB)\subseteq r (r (A)B).
\end{equation*}
The converse inclusions hold if $X$ is local.
\end{lemma}
The proofs have again been checked with Isabelle and can be found in
Appendix~\ref{A:three}. Weak locality holds in any weakly local modal
semiring and quantale (which satisfy export axioms, see
Section~\ref{S:modal-quantales}). 

The results of the previous two lemmas can be summarised as follows.

\begin{theorem}\label{P:lr-pow-lifting}
  Let $X$ be an $\ell r$-multimagma. 
\begin{enumerate}
\item Then $(\Pow X,\subseteq,\odot,E,\dom,\cod)$ is a boolean modal
  prequantale in which $\dom (A) = \ell(A)$, $\cod (A)= r (A)$, for
  all $A\subseteq X$, and the complete boolean algebra is atomic.
\item It is a weakly local modal quantale if $X$ is an $\ell r$-multisemigroup.
\item It is a modal quantale if $X$ satisfies locality.
\end{enumerate}
\end{theorem}
\begin{proof}
  We have derived the respective variants of modal prequantale and
  quantale axioms in Lemmas~\ref{P:lr-magma-lift} and
  \ref{P:lr-semigroup-lift} for $\ell r$-magmas, $\ell r$-semigroups
  and local $\ell r$-multisemigroups. They hold in addition to the
  boolean prequantale and quantale axioms lifted via
  Proposition~\ref{P:lr-pow} and the remark following it.
\end{proof}

This shows the particular role of weak locality and locality in the
three stages of lifting. The construction shown is one direction of
the well known J\'onsson-Tarski duality between relational structures
and boolean algebras with operators~\cite{JonssonT51}, which
generalises to categories of relational structures and boolean
algebras with operators~\cite{Goldblatt89}.
Theorem~\ref{P:lr-pow-lifting} is an instance of this duality. Like in
modal logic, there are correspondences between relational structures
and boolean algebras with operators. The identities lifted in
Lemma~\ref{P:lr-magma-lift} and \ref{P:lr-semigroup-lift} are one
direction of these. They are further investigated in a more general
setting in Section~\ref{S:correspondences}.

\begin{example}[Modal Powerset Quantales over $\ell
  r$-Semigroups]\label{ex:mp-quantales}~
  \begin{enumerate}
  \item Any category as a local partial $\ell r$-semigroup can be
    lifted to a modal powerset quantale. It is boolean and has the
    arrows of the category as atoms. The domain algebra is the entire
    boolean subalgebra below the unit of the quantale, the set of all
    objects of the category (or the identity arrows). In this sense, a
    modal algebra can be defined over any category. 

  \item As an instance, in the modal powerset quantale over the pair
    groupoid on $X$, that is, the modal quantale of binary relations,
    the domain and codomain elements are the relational domains and
    codomains of relations mentioned in Example~\ref{ex:rel-dom-cod}.
    Domain and codomain elements are precisely the subidentity
    relations below $\mathit{Id}_X$.  In the associated matrix
    algebra, these correspond to (boolean-valued) identity matrices
    and further to predicates, cf.\ Example
    \ref{ex:lr-mgs}(\ref{ex:matrices}). The domain and codomain
    operators allow constructing predicate transformers and algebraic
    variants of dynamic logics, with applications in program
    verification; see Section~\ref{S:dualities}.

  \item In this and the next example, locality at powerset-level
    fails.  Recall that the partial $\ell r$-semigroup in the broken
    monoid (Example \ref{ex:lr-mgs}(\ref{ex:broken-mon-bis})) is only
    weakly local. The powerset quantale is only weakly local as
    well. To check this, we simply replay the non-locality proof for
    the partial $\ell r$-semigroup with $A=\{a\}$:
    $\dom (AA) = \dom (\emptyset) = \emptyset \subset \{1\} = \dom (A
    \{1\}) = \dom (A \dom (A))$.
    Locality of codomain is ruled out by duality.
  \item Locality of domain and codomain of the powerset algebra over
    the non-local partial abelian monoid of heaplets is ruled out with
    singleton sets as in the previous example, using the non-locality
    argument from Example~\ref{ex:lr-mgs}(\ref{ex:heaplets2}).  The
    powerset quantale of heaplets is the assertion quantale of
    algebraic separation logic with separating conjunction as
    quantalic composition. In separation logic, another notion of
    locality is related a frame rule of its inference system. It is
    completely unrelated to the notion used in this article.\qed
  \end{enumerate}
\end{example}

In these examples, antidomain and anticodomain operations can be
defined along the lines of Section~\ref{S:modal-quantales}.  Most of
the models of domain and modal semirings considered previously are in
fact powerset structures lifted from
categories. Theorem~\ref{P:lr-pow-lifting} yields a uniform
construction recipe for all of them. Further examples of modal
prequantales, weakly local modal quantales and modal quantales lifted
from underlying $\ell r$-structures can be found in
Section~\ref{S:examples}.

The two final examples of this section show that Axioms (\ref{eq:3c})
and (\ref{eq:D3c}) for function systems (shown in
Appendix~\ref{A:one}) do not lift to powersets.

\begin{example}\label{ex:no-twisted}~
\begin{enumerate}
\item The category $1\stackrel{a}{\to}2$, also known as walking arrow,
  forms the partial local $\ell r$-semigroup with elements
  $X=\{1,a,2\}$, $\ell$ and $r$ defined by
  $\ell (1) = r (1) = 1 =\ell (a)$ and $\ell (2) = r (2) =2 = r (a)$
  and composition $11=1$, $1a=a= a2$ and $22=2$.  Then, for
  $A=\{1,a\}$ and $B=\{2\}$,
  \begin{equation*}
    A\cdot \dom(B) = A\cdot B = \{a\} \subset A = \{1\}\cdot A =
    \dom(A\cdot B)\cdot A
  \end{equation*}
  refutes (\ref{eq:3c}) in $\Pow X$.  The opposite law (\ref{eq:D3c})
  for $\cod$ is refuted by a dual example.
\item Axiom (\ref{eq:3c}) also fails in modal quantales of
  relations. Encoding the walking arrow on the set $X=\{a,b\}$ using
  the relations
\begin{equation*}
\begin{tikzcd}[column sep=10mm,row sep=3mm]
a\arrow[->,r,"R"]\arrow[->,out=140,in=220,loop,swap,"R"] & b\arrow[->,out=40,in=320,loop,"S"]
\end{tikzcd}
\end{equation*}
yields $R\dom(S) = RS= \{(a,b)\} \subset R = \{(a,a)\}R = \dom(RS)R$
in $\Pow (X\times X)$.  The expression $\dom(RS)=\dom(R\dom(S))$
models the relational preimage of $\dom(S)$ under $R$.  Obviously,
executing $R$ from all those inputs that may lead into $\dom(S)$ and
restricting the outputs of $R$ to $\dom(S)$ is only the same when $R$
is a function.

A dual example for $\cod$ and (\ref{eq:D3c}) uses the converses of $R$
and $S$.  From the discussion above it is evident that (\ref{eq:D3c})
does not hold for general functions, yet it does for monos in
$\mathsf{Rel}$.\qed
\end{enumerate}
\end{example}

Hence we remain within the realm of modal quantales as opposed to
function systems.  Weak variants of Schweizer and Sklar's axioms
(\ref{eq:3c}) and (\ref{eq:D3c}),
$\alpha\dom (\beta) \le\dom (\alpha\beta) \alpha$ and
$\cod (\alpha) \beta\le \beta \cod (\alpha\beta)$, can be derived in
any modal semiring, as already mentioned, yet they need not hold in
modal prequantales. The equational $\ell r$-multisemigroup variants of
Axioms (\ref{eq:3c}) and (\ref{eq:D3c}) do not lift to powersets.  A
well known theorem by Gautam~\cite{Gautam57} shows that identities
lift to the powerset level if and only if all variables occur on each
side of the identity (or else the two sides are identical). Neither
the generalisations of (\ref{eq:3c}) and (\ref{eq:D3c}) satisfy this
condition, nor their equational specialisations.

\begin{remark}
  The functor $G:\mathbf{Set}\to \mathbf{Rel}$ that maps functions to
  their graph associates convolutions of (graphs of) functions with
  function composition $G(f)\odot G(g)= G(f; g)=G(g\circ f)$ and
  distributes over $\dom$:
  $G(\dom(f))= \textit{Id}_{\dom(G(f))} = \dom(G(f)) $.  Axiom
  (\ref{eq:3c}) can then be derived in the convolution algebra:
  \begin{equation*}
    G(f)\odot \dom(G(g)) = G(f;\dom(g)) =
    G(\dom(f;g);f) = 
    \dom(G(f)\odot G(g))\odot G(f)
  \end{equation*}
  Yet this depends on the specific form of convolution for (graphs of)
  functions and thus goes beyond the general lifting by convolution.
  The argument for (\ref{eq:D3c}) is analogous.
\end{remark}

An asymmetry between $\ell r$-multisemigroups and modal quantales
remains. While the domain axioms for quantales are purely equational,
those for $\ell r$-multisemigroups are based on the implication
$D_{xy}\Rightarrow r(x)=\ell(y)$, and we do not know an equational
axiomatisation for this class. We leave this as an open question.


\section{Modal Convolution Quantales}\label{S:conv-alg}

In this section we prove one of the main theorems in this article, and
perhaps its most useful one.  It refines Proposition~\ref{P:lr-pow} as
a powerset lifting. Here we aim at a similar refinement of
Theorem~\ref{P:lr-conv} as a lifting to functions valued in domain
quantales.  

Once again we aim to expose the conditions on the algebras used in the
lifting, including modal prequantales and weakly local modal
quantales. Yet first we need to generalise the definition of domain
and codomain at the powerset level to make it suitable for convolution
quantales. For every $\ell r$-multimagma $X$, modal prequantale $Q$
and functions $f,g:X\to Q$ we define the operations $\Dom$ and $\Cod$
by
\begin{equation*}
  \Dom(f) = \Sup_{x\in X} \dom(f(x))\cdot \delta_{\ell(x)}\qquad
  \text{ and }\qquad \Cod(f)= \Sup_{x\in X} \cod(f(x))\cdot
  \delta_{r(x)}. 
\end{equation*}
These definitions imply that $\Dom(f)(x)= \Cod(f)(x)= \bot$ for
$x\notin E$.  Restricted to functions $E\to Q$, they are equivalent to
$\Dom(f)\circ \ell= \dom\circ f$ and $\Cod(f)\circ r= \cod\circ f$.

\begin{theorem}\label{P:lr-lifting}
  Let $(X,\mcomp,\ell,r)$ be an $\ell r$-multimagma and $Q$ a modal
  prequantale.
 \begin{enumerate}
\item Then $(Q^X,\le,\ast,\id_E,\Dom,\Cod)$ is a
  modal prequantale.
\item It is a weakly local modal quantale if $X$ is an $\ell
  r$-multisemigroup and $Q$ a weakly local modal quantale.
\item It is a modal quantale if $X$ is also local and $Q$ a modal
  quantale.
  \end{enumerate}
\end{theorem}

\begin{proof}
  Relative to Theorem~\ref{P:lr-conv} and the remark following it we
  need to check the domain and codomain axioms as well as the
  compatibility axioms. We show proofs up-to duality. We point out
  where an $\ell r$-multimagma $X$ together with a modal prequantale
  $Q$ or an $\ell r$-multisemigroup together with a weakly local modal
  quantale suffices for the lifting: this is the case in (1)-(5)
  below.  Here it is convenient to view $\alpha\cdot \delta_x$ as an
  element of the convolution algebra or an associated (imaginary)
  $Q$-module.
\begin{enumerate}
\item First we show sup-preservation in $Q^X$, assuming the
  corresponding law in $Q$:
\begin{align*}
  \Dom \left(\Sup F\right)  
&= \Sup_{x} \dom \left(\Sup \{f(x) \mid f\in F\}\right) \cdot
  \delta_{\ell(x)}\\
&= \Sup\left\{\Sup_{x} \dom(f(x))\cdot \delta_{\ell(x)} \mid f\in F\right\}
  \\
&=\Sup\{\Dom(f)\mid f \in F\}.
\end{align*}
Preservation of binary sups and domain strictness, $\Dom(\lambda x.\
\bot) =\bot$, then follow. 

\item For the first compatibility axiom in $Q^X$, assuming the
  corresponding laws in $X$ and $Q$,
\begin{align*}
  (\Dom\circ \Cod)(f) 
&= \Sup_{x} \dom\left(\Sup_{y}\cod\left(f(y)\right)\cdot
  \delta_{r(y)}(x)\right)\cdot \delta_{\ell(x)}\\
&= \dom\left(\Sup_{y}\cod(f(y))\right)\cdot \delta_{\ell(r(y))}\\
&= \Sup_{y}\dom(\cod(f(y)))\cdot \delta_{\ell(r(y))}\\
&= \Sup_{y}\cod(f(y))\cdot \delta_{r(y)}\\
&= \Cod(f).
\end{align*}

\item \label{en:lr-lifting.pf.subid} For the domain subidentity axiom
  in $Q^X$, assuming the corresponding law in $Q$,
\begin{equation*}
\Dom (f) = \Sup_{x}\dom (f (x))\cdot \delta_{\ell(x)} \le \Sup_{x
  \in E} 1\cdot \delta_x  = \id_E.
\end{equation*}

\item For domain absorption in $Q^X$, assuming the corresponding laws
  in $X$ and $Q$, 
\begin{align*}
  \Dom (f) \ast f
&= \Sup_{w,x,y} \left(\Sup_{z}\dom(f(z))\cdot
  \delta_{\ell(z)}(x)\right)\cdot f(y) \cdot [w \in x\odot y]\cdot \delta_w\\
&= \Sup_{w,y,z} \dom(f(z))\cdot f(y) \cdot [w\in \ell(z)\odot
  y]\cdot \delta_w\\
&\ge \Sup_{w,y} \dom(f(y))\cdot f(y) \cdot [w\in \ell(y)\odot y]\cdot \delta_w\\
&= \Sup_{y} f(y) \cdot \delta_{y}\\
&= f,
\end{align*}
and $\Dom(f) \ast f\le f$ follows from \eqref{en:lr-lifting.pf.subid}.

\item For domain export in $Q^X$, assuming the corresponding laws in
  $X$ and $Q$,
\begin{align*}
  \Dom (\Dom (f) \ast g)
&= \Sup_{x} \dom\left(\Sup_{y,z}\left(\Sup_w\dom(f(w))\cdot \delta_{\ell(w)}(y)\right)\cdot g(z)\cdot
  [x\in y\odot z]\right) \cdot \delta_{\ell(x)}\\
&=\Sup_{x,z,w} \dom(\dom(f(w)) \cdot g(z))\cdot 
  [\ell(x)\in \ell(\ell(w)\odot z)]\\
&=\Sup_{x,z,w} \dom(f(w)) \cdot \dom(g(z))\cdot 
  [\ell(x)\in\ell(w)\odot \ell(z)]\cdot \delta_{\ell(x)}\\
&=\left(\Sup_w\dom(f(w))\cdot \delta_{\ell(w)}\right)\ast
  \left(\Sup_z\dom(g(z))\cdot \delta_{\ell(z)}\right)\\
&=\Dom(f)\ast \Dom(g).
\end{align*}

This is not a domain quantale axiom, but a domain axiom for
prequantales, as already mentioned.

\item For weak domain locality in $Q^X$, assuming the corresponding
  laws in $X$ and $Q$, 
\begin{align*}
\Dom (f \ast g)
&= \Sup_x \dom\left(\Sup_{y,z}f(y)\cdot g(z) \cdot [x\in y\odot z]\right)
  \cdot \delta_{\ell(x)}\\
 &= \Sup_{x,y,z}\dom(f(y)\cdot g(z)) \cdot
  [\ell(x)\in \ell(y\odot z)]\\
 &\le \Sup_{x,y,z}\dom(f(y)\cdot \dom(g(z))) \cdot
   [\ell(x)\in \ell(y\odot \ell(z))]\\
 &=\Sup_x \dom\left(\Sup_{y,z}f(y)\cdot \left(\Sup_w\dom(g(w))\cdot
   \delta_{\ell(w)}(z)\right) \cdot [x\in y\odot z]\right)\cdot \delta_{\ell(x)}\\
 &= \Dom (f \ast \Dom (g)).
\end{align*}

\item For domain locality, we replay the weak locality proof with
  equations. This requires locality in $X$ and $Q$.\qedhere
\end{enumerate}
\end{proof}

\begin{remark}
  Theorem~\ref{P:lr-lifting} comprises the case where $Q$ is merely a
  (unital) quantale, that is, $\dom$ and $\cod$ map $\bot$ to $\bot$
  and all other elements to $1\neq \bot$ (or another fixed element
  $\alpha\neq \bot$). A further specialisation leads to the quantale $2$ of
  booleans, and hence Theorem~\ref{P:lr-pow-lifting}.
\end{remark}

Full modal correspondence results for $Q$-valued functions are more
complicated than for powersets. They are investigated in the next
section. Finally, examples for all three cases of
Theorem~\ref{P:lr-lifting}, beyond mere powerset liftings, are
presented in Section~\ref{S:examples}.

Next we return to our two running examples.
\begin{example}[Modal Convolution Quantales over $\ell
  r$-Semigroups]\label{ex:modal-convolution-quantales}~
  \begin{enumerate}
  \item In the construction of category quantales, the
    $\ell r$-structure of the underlying category lifts to the modal
    structure of the convolution quantale. In this sense, a
    $Q$-valued modal algebra can be defined over any category.
  \item As an instance, for any modal quantale $Q$, the convolution
    algebra $Q^X$ over the pair groupoid $X$ forms a modal quantale.
    In this algebra,
    \begin{align*}
      \Dom(f)(a,b) &= \Sup_c
      \dom(f(a,c))\cdot
      \delta_a(b)=\dom\left(\Sup_cf(a,c)\right)\cdot \delta_a(b),\\
\Cod(f)(a,b) &= \Sup_c\cod(f(c,a))\cdot \delta_a(b)
               =\cod\left(\Sup_cf(c,a)\right)\cdot \delta_a(b).
\end{align*}
The subalgebra of weighted domain and codomain elements is then the
algebra of all weighted elements below the identity relation.  These
can be identified with $Q$-valued predicates. For $Q=2$, this reduces
to the standard definitions of $\dom$ and $\cod$ in the quantale of
binary relations.  When $Q$-valued relations are viewed as possibly
infinite dimensional matrices, domain and codomain elements correspond
to diagonal matrices with values given by domain and codomain elements
below $1$ in $Q$ along the diagonal according to the formulas above
and $\bot$ everywhere else.  For instance,
\begin{align*}
  \Dom \left( 
    \begin{pmatrix}
      \alpha & \beta\\
\gamma & \delta
    \end{pmatrix}
\right)
&=
\begin{pmatrix}
  \dom(\alpha\sup \beta) & \bot\\
\bot & \dom(\gamma \sup\delta)
\end{pmatrix},\\
 \Cod \left( 
    \begin{pmatrix}
      \alpha & \beta\\
\gamma & \delta
    \end{pmatrix}
\right)
&=
\begin{pmatrix}
  \cod(\alpha\sup \gamma) & \bot\\
\bot & \cod(\beta\sup\delta)
\end{pmatrix}.
\end{align*}

This can be seen as a refinement of the matrix theories of Example
\ref{ex:lr-mgs}(\ref{ex:matrices}).\qed
  \end{enumerate}
\end{example}

The matrix example does not actually require a quantale as value
algebra: in a finite weighted relation, the summation used in
relational composition, that is, matrix multiplication, is over a
finite set and can therefore be represented by a finite
supremum.  In the absence of domain and codomain in the value
algebra, even a semiring can be used. 

More generally, we can require that the multioperation $\mcomp$ of the
$\ell r$-multisemigroup satisfies a finite decomposition property.
Here are two classical examples beyond matrices.
\begin{itemize}
\item Sch\"utzenberger and Eilenberg's approach to weighted formal
  languages~\cite{book/DrosteKV09} generalises language product to a
  convolution of functions $\Sigma^\ast\to S$ from the free monoid
  $\Sigma^\ast$ over the finite alphabet $\Sigma$ into a semiring $S$,
  as mentioned in the introduction. A semiring suffices in the
  convolution because the number of prefix/suffix pairs into which any
  finite word in $\Sigma^\ast$ can be split is obviously finite. The
  resulting convolution algebra is again a semiring.
\item Rota's incidence algebra~\cite{Rota64} of functions $P\to R$
  from a poset $P$ to a commutative ring $R$ requires $P$ to be
  locally finite, that is, every closed segment
  $[x,y]=\{z\mid x\le z\le y\}$ must be finite. The incidence algebra
    is then an associative algebra in which the sup of the convolution
    is replaced by a summation in $R$. See also
    Section~\ref{S:examples}.
  \end{itemize}
  Alternatively, in groups rings, it is usually assumed that functions
  $G\to R$ have finite support, yet we focus on the property used for
  matrices, formal powerseries and incidence algebras. 

  A multimagma $(X,\mcomp)$ is \emph{finitely decomposable} if for
  each $x\in X$ the set $\{(y,z)\mid x \in y\mcomp z\}$ is finite. The
  following variant of Theorem~\ref{P:lr-conv} is then immediate from
  corresponding results for relational monoids~\cite{CranchDS20a}.

\begin{theorem}\label{P:interchange-semiring-correspondence}
  If $X$ is a finitely decomposable $\ell r$-multisemigroup and $(S,+,\cdot,0,1)$ a
  semiring, then the convolution algebra $S^X$ is a semiring.
\end{theorem}
A direct proof requires verifying that all sups in the proof of
Theorem~\ref{P:lr-conv} remain finite if $X$ is finitely decomposable.
The operation $+$ of semirings corresponds to $\sup$ in quantales, $0$
corresponds to $\bot$.

This result easily extends to domain semirings. Formally, a
\emph{domain semiring}~\cite{DesharnaisS11} is a semiring
$(S,+,\cdot,0,1)$ equipped with a domain operation $\dom:S\to S$ that
satisfies the same domain axioms as those of domain quantales,
replacing $\bot$ by $0$ and $\sup$ by $+$. Every domain semiring is
automatically additively idempotent, and hence the relation $\le$,
defined as $\lambda x,y.\ x+y = y$, is a partial order.

A \emph{modal semiring} is a domain semiring and a codomain semiring,
defined like for quantales by opposition, which satisfy the
compatibility conditions $\cod\circ\dom = \dom$ and
$\dom\circ \cod = \cod$.

Verifying finiteness of sups in the proof of
Theorem~\ref{P:lr-lifting} then extends
Theorem~\ref{P:interchange-semiring-correspondence} as follows.

\begin{proposition}\label{P:interchange-modal-semiring-correspondence}
  If $X$ is a finitely decomposable local $\ell r$-multisemigroup
  and $S$ a modal semiring, then the convolution algebra $S^X$ is a
  modal semiring.
\end{proposition}
The small sup-preservation properties of
Lemma~\ref{P:lr-magma-lift}(4) are needed in the proof.  The result
can obviously be adapted to $\ell r$-multimagmas and modal
presemirings, and to $\ell r$-multisemigrou ps and weakly modal
semirings, as in previous sections.  These results cover many
examples, as we shall see in Section~\ref{S:examples}.


\section{From Liftings to Correspondences}\label{S:correspondences}

Section~\ref{S:conv-alg} has presented lifting results from
$\ell r$-multimagmas to modal convolution prequantales. These yield
one direction of a triangle of modal correspondences between
$\ell r$-magmas $X$, modal prequantales $Q$ and modal prequantales
$Q^X$, shown in Figure~\ref{fig:correspondence-triangle}.  In the
relational setting, the triangular correspondences
from~\cite{CranchDS20a} show that certain properties in any two of a
relational magma $X$, a prequantale $Q$ and a prequantale $Q^X$ induce
corresponding properties in the remaining algebra, for instance
associativity.  They include the following results for units,
translated to multimagmas by isomorphism.
\begin{proposition}\label{P:lr-correspondence2}
Let $X$ be a multimagma and $Q$ a prequantale, not necessarily unital.
\begin{enumerate}
\item If the prequantale $Q^X$ is unital and $1\neq \bot$ in $Q$, then
  $X$ is an $\ell r$-multimagma.
\item If the prequantale $Q^X$ is unital and $X$ an
  $\ell r$-multimagma, then $Q$ is unital.
\end{enumerate}
\end{proposition}
\begin{proof}
  The results are known for relational magmas \cite[Proposition
  4.1]{CranchDS20a} and thus hold for
  multimagmas. Proposition~\ref{P:lr-semigroup-rel-monoid} translates
  them to $\ell r$-multimagmas.
\end{proof}
The restriction $\bot\neq 1$ in Proposition~\ref{P:lr-correspondence2}
is very mild as its negation would imply $Q=\{ \bot\}$.

We now complete the triangle in
Figure~\ref{fig:correspondence-triangle} by refining
Proposition~\ref{P:lr-correspondence2} to variants of modal quantales.
Relative to Proposition~\ref{P:lr-correspondence2} we only consider the
$\ell r$-axioms and the modal axioms for prequantales and quantales
while assuming the multimagma and prequantale axioms. The
constructions use a technique from~\cite{CranchDS20a} which in turn has
been adapted from representation theory in algebra.

In the following theorems, we tacitly assume that $\ell$, $\dom$ and
$\Dom$, and similarly $r$, $\cod$ and $\Cod$, are related as in
Theorem~\ref{P:lr-lifting}.  In particular, therefore
\begin{equation*}
\Dom(\alpha\cdot \delta_x) = \Sup_y \dom(\alpha\cdot
\delta_x(y))\delta_{\ell(y)} = \dom(\alpha) \cdot
\delta_{\ell(x)}\quad\text{ and }\qquad \Dom(\delta_x) =\delta_{\ell(x)}
\end{equation*}
with dual laws for $\Cod$, and
\begin{equation*}
  (\alpha\cdot \delta_x\ast \beta\cdot \delta_y)(z)= \Sup_{u,v}
  \alpha\cdot \delta_x(u)\cdot \beta\cdot 
  \delta_y(v)\cdot [z\in u\odot v]= \alpha\cdot \beta\cdot [z\in x\odot y]
\end{equation*}
and thus in particular $(\delta_x\ast \beta\cdot \delta_y)(z)=[z\in x\odot y]$.

Our first lemma link properties of $X$ with those of $Q$ and $Q^X$
through $\delta$-functions.

\begin{lemma}\label{P:lr-correspondence-aux}
  Let $X$ be a multimagma with functions $\ell,r:X\to X$, let $Q$ and
  $Q^X$ be prequantales with functions $\dom,\cod:Q\to Q$ and
  $\Dom,\Cod:Q^X\to Q^X$. Then, for all $\alpha,\beta\in Q$ and $x,y,z\in X$,
  \begin{enumerate}
  \item
    $\Dom(\alpha\cdot \delta_x) \ast (\alpha\cdot \delta_x) = \Sup_y
    \dom(\alpha)\cdot \alpha\cdot [y \in \ell(x)\mcomp x]\cdot \delta_y$,
 \item $\Dom(\Dom(\alpha\cdot \delta_x)\ast (\beta\cdot \delta_y) =
    \Sup_z\dom(\dom(\alpha)\cdot \beta)\cdot  [z \in \ell(\ell(x)
    \odot y)]\cdot \delta_{\ell(z)}$,
\item $\Dom(\alpha\cdot \delta_x)\ast \Dom(\beta\cdot \delta_y) = \Sup_z\dom(\alpha)\cdot
  \dom(\beta)\cdot [z\in \ell(x)\odot \ell(y)]\cdot \delta_z$,
  \item
    $\Dom((\alpha\cdot \delta_x) \ast \Dom\left(\beta\cdot \delta_y)\right) = \Sup_z\dom(\alpha\cdot
    \dom(\beta))\cdot [ z \in \ell(x \mcomp \ell(y))]\cdot \delta_{\ell(z)}$,
\item $\Dom((\alpha\cdot \delta_x) \ast (\beta\cdot \delta_y)) = \Sup_z\dom(\alpha\cdot
  \beta)\cdot [z \in \ell(x \mcomp y)]\cdot \delta_{\ell(z)}$.
\item Corresponding properties hold for $\Cod$, $\cod$ and $r$. 
  \end{enumerate}
\end{lemma}
\begin{proof} We write $\delta^\alpha_x$ instead of
  $\alpha\cdot \delta_x$ and drop multiplication symbols wherever
  convenient.
  \begin{enumerate}
  \item
    $\Dom(\delta^\alpha_x) \delta^\alpha_x
    =\Sup_{y,u,v}\delta^{\dom(\alpha)}_{\ell(x)}(u) \delta^\alpha_x(v)
    [y\in uv]\delta_y =\Sup_y\dom(\alpha)\alpha [y \in \ell(x)x]\delta_y$.
  \item 
    \begin{align*}
      \Dom(\Dom(\delta^\alpha_x)\delta^\beta_y) 
&=\Sup_z \dom\left(\Sup_{v,w}\delta^{\dom(\alpha)}_{\ell(x)}(v)
  \delta^\beta_y(w)[z\in vw]\right)\delta_{\ell(z)}\\
&=\Sup_z \dom(\dom(\alpha)\beta)[z\in \ell(x)y])
  \delta_{\ell(z)}\\
&=\Sup_z\dom(\dom(\alpha)\beta)[z\in \ell(\ell(x)y)]\delta_z.
    \end{align*}
\item 
$\Dom(\delta^\alpha_x)\Dom(\delta^\beta_y)
=\Sup_{z,v,w}
  \delta^{\dom(\alpha)}_{\ell(x)}(v)
  \delta^{\dom(\beta)}_{\ell(y)}(w)[z\in vw]\delta_z
=\Sup_z\dom(\alpha)\dom(\beta)[z\in \ell(x)\ell(y)]\delta_z$. 
\item 
\begin{align*}
  \Dom\left((\delta^\alpha_x) \Dom(\delta^\beta_y)\right)
&=\Sup_z \dom\left(\Sup_{u,v}\delta^\alpha_x(u)\delta^{\dom(\beta)}_{\ell(y)}(v)[z \in uv]\right)\delta_{\ell(z)}\\
&=\Sup_z \dom(\alpha \dom(\beta)[z \in x\ell(y)])
  \delta_{\ell(z)}\\
&= \Sup_z\dom(\alpha\dom(\beta))[ z \in \ell(x\ell(y))]\delta_{\ell(z)}.
\end{align*}
\item
\begin{align*}
  Dom(\delta^\alpha_x\delta^\beta_y)
  &=\Sup_z \dom\left(\Sup_{v,w}\delta^\alpha_x(v) \delta^\beta_y(w)
    [z\in vw]\right)\delta_{\ell(z)}\\
  &=\Sup_z \dom\left(\alpha\beta
    [z\in xy]\right)\delta_{\ell(z)}\\
  &= \Sup_z\dom(\alpha\beta)[z \in \ell(xy)]\delta_{\ell(z)}.
\end{align*}
\item Proofs for $\Cod$, $\cod$ and $r$ are dual. \qedhere
 \end{enumerate}
\end{proof}

The following statements add structure to
Proposition~\ref{P:lr-correspondence2}(1). They expose the laws in $Q$
and $Q^X$ needed to derive the $\ell r$-multimagma and
$\ell r$-multisemigroup axioms, with and without locality.  The
$\ell r$-multisemigroup structure arises from of~\cite[Corollary
4.7]{CranchDS20a} for relational semigroups together with
Proposition~\ref{P:lr-semigroup-rel-monoid}, which translates it to
$\ell r$-semigroups. Proving them requires very mild assumptions on
$X$ and $Q$ or $Q^X$.

\begin{proposition}\label{P:assoc-correspondence}
  If $Q^X$ and $Q$ are quantales and $1\neq \bot$ in $Q$, then $X$ is
  an $\ell r$-multisemigroup.
\end{proposition}

The $\ell r$-multisemigroup $X$ is therefore completely determined by
elements below $\id_E$, more specifically, functions
$\delta_{\ell(x)}$ and their relations to the elements in $X$.  We
calculate the absorption law for $\ell$ explicitly as an example of
the technique used: With Lemma~\ref{P:lr-correspondence-aux}(1),
$\delta_{x} = \Dom(\delta_x)\ast \delta_x= \Sup_y[y\in
\ell(x)\mcomp x]\delta_y$.
Hence $\ell(x)\mcomp x = \{x\}$ whenever the corresponding domain
absorption law holds in $Q^X$ and $1\neq \bot$ in $Q$. The fact that
$\Dom$ appears in the calculation does not go beyond
Proposition~\ref{P:assoc-correspondence}:
$\Dom(\delta_x) = \delta_{\ell(x)}$, which is below $\id_E$ in $Q^X$.

The next statement adds locality to the picture.
  
\begin{theorem}\label{P:lr-correspondence2-thm}
  Let $Q^X$ and $Q$ be modal quantales, with $1\neq \bot$ in $Q$. Then
  $X$ is a local $\ell r$-multisemigroup.
\end{theorem}
\begin{proof}
  It remains to consider locality.
  Lemma~\ref{P:lr-correspondence-aux} yields
  $\Dom(\delta_x \ast \Dom\left(\delta_y)\right)
  = \Sup_z[z \in \ell(x\odot \ell(y))]\delta_{\ell(z)}$
  and
  $\Dom(\delta_x \ast \delta_y) = \Sup_z[z \in
  \ell(x\odot y)]\delta_{\ell(z)}$.
  Therefore
  $\Dom(\delta_x \ast \Dom\left(\delta_y)\right)(z)
  =\Dom(\delta_x \ast \delta_y)(z)$
  implies $\ell (x\odot \ell(y)) = \ell(x\odot y)$.
\end{proof}

Finally, we turn to the correspondences from $X$ and $Q^X$ to $Q$. We
first consider modal axioms in $Q$ that do not depend on
$\ell r$-multisemigroups.  Similarly to the assumption that $1\ne
\bot$ in $Q$ above, we need to assume the existence of certain
elements in $X$.

\begin{theorem}\label{T:correspondence-X+QX->Q}
  Let $X$ be an $\ell r$-multimagma in which there exist
  $x,y,z,w\in X$, not necessarily distinct, such that
  $\ell(x)\mcomp \ell(y)\neq \emptyset$ and $z\mcomp w\neq \emptyset$,
  let $Q$ be a prequantale and $Q^X$ a modal prequantale such that
  $\id_E\neq \bot$.
  \begin{enumerate}
  \item Then $Q$ is a modal prequantale.
\item It is a weakly modal quantale if $X$ is an $\ell
  r$-multisemigroup and $Q^X$ a weakly local quantale.
\item It is a modal quantale if $X$ is also local and $Q^X$ a modal
  quantale. 
  \end{enumerate}
\end{theorem}
\begin{proof}
  Note that by definition $\emptyset\neq E\subseteq X$. We verify the
  modal quantale axioms in $Q$, using in each case the corresponding
  axiom in $X$ and $Q^X$. Suppose that $X$ is an $\ell r$-magma and
  $Q^X$ a modal prequantale with $\id_E\neq \bot$.
\begin{itemize}
\item For domain absorption, using
  Lemma~\ref{P:lr-correspondence-aux}(1) with
  $\ell(x)\odot x = \{x\}$, 
  \begin{equation*}
    \dom(\alpha)\cdot \alpha =
    (\Dom(\delta^\alpha_x)\ast (\delta^\alpha_x))(x) = \delta^\alpha_{x}(x) =\alpha.
\end{equation*}
\item For domain export, using Lemma~\ref{P:lr-correspondence-aux}(2)
  and (3) with $z\in \ell(\ell(x)\odot y)=\ell(x)\odot
  \ell(y)$,
  \begin{equation*}
    \dom(\dom(\alpha)\cdot \beta) = \Dom(\Dom(\delta^\alpha_x)\ast \delta^\beta_y)(z) =
    (\Dom(\delta^\alpha_x)\ast \Dom(\delta^\beta_y))(z) = \dom(\alpha)\cdot
    \dom(\beta). 
  \end{equation*}
\item   For the domain subidentity axiom, $\dom(\alpha) = 
  \Dom(\delta^\alpha_{x})(\ell(x))\le \id_E(\ell(x)) = 1$.
\item For domain strictness,
  $\dom(\bot) =\Dom(\delta^\bot_x)(\ell(x))=\Dom(\bot)(\ell(x))=\bot$.
\item For binary sup-preservation of domain, 
\begin{align*}
  \dom(\alpha\sup \beta) 
&=\Dom(\delta^{\alpha\sup\beta}_x)(\ell(x))\\
&=\Dom(\delta^\alpha_x\sup \delta^\beta_x)(\ell(x))\\
&=\Dom(\delta^\alpha_x)(\ell(x))\sup \Dom(\delta^\beta_x)(\ell(x))\\
&= \dom(\alpha)\sup \dom(\beta).
\end{align*}
\item For weak domain locality, using
  Lemma~\ref{P:lr-correspondence-aux}(2) and (3) with $z\in \ell(x\odot y)
  \subseteq \ell(x\odot \ell(y))$,
\begin{align*}
\dom(\alpha\cdot \beta)
= \Dom(\delta^\alpha_x\ast \delta^\beta_y)(z)
\le \Dom((\delta^\alpha_x) \ast \Dom(\delta^\beta_y))(z)
=   \dom(\alpha\cdot \dom(\beta)).
\end{align*}
\item For domain locality, replay the previous proof with identities
  instead of inequalities.
\item  For compatibility, $\dom(\cod(\alpha)) = \Dom (\Cod(\delta^\alpha_x))(\ell(r(x))) =
    \Cod(\delta^\alpha_x)(r(x)) = \cod(\alpha)$. 
 \end{itemize}
The remaining proofs follow by duality. \qedhere
\end{proof}


\section{Examples}\label{S:examples}

In this section we list additional examples of modal convolution
quantales. We start with those that come from categories. 

\begin{example}[Modal Convolution Quantales over Path Categories of
  Digraphs (Quivers)]\label{ex:traces}~
\begin{enumerate} 
\item \label{ex:path-category} A digraph or quiver is a structure
  $K$ consisting of a set $V_K$ of vertices, a set $E_K$ of edges and
  source and target functions $s, t: E_K\to V_K$.  The path category
  of $K$~\cite{MacLane98} has elements $V_K$ as objects and sequences
  $(v_1,e_1,v_2,\dots, v_{n-1},e_{n-1},v_n):v_1\to v_n$, in which
  vertices and edges alternate, as arrows.  Composition
  $\pi_1\cdot \pi_2$ of $\pi_1:v_3\to v_4$ and $\pi_2:v_1\to v_2$ is
  defined whenever $v_2=v_3$, and it concatenates the two paths while
  gluing the common end $v_2=v_3$. Sequences $(v)$ of length $1$ are
  identities. Path categories become local partial $\ell r$-semigroups
  if we introduce an edge $i_v$ for every vertex $v$ and functions
  $\ell(e)=i_{s(e)}$ and $r(e)= i_{t(e)}$ for every edge $e$, and it
  is common to define path algebras of quivers this way as sequences
  of arrows. By Theorem~\ref{P:lr-lifting}, the convolution algebra or
  category algebra over the path category of any digraph with values
  in $Q$ is a modal quantale for any modal quantale $Q$. All paths are
  finitely decomposable, hence we can replace $Q$ by a modal semiring.
 
\item A special path category is generated by the one-point quiver
  with $n$ arrows. It represents the free monoid with $n$
  generators. The $\ell r$-structure and hence the modal structure is
  then trivial. Lifting along Theorem~\ref{P:lr-lifting} yields the
  quantale or semiring of weighted languages.
\item Forgetting edges represents paths as sequences of vertices. All
  lifting results transfer.  
\item Forgetting the internal structure of paths and keeping only
  their ends, brings us back to the pair groupoid and weighted binary
  relations.\qed
\end{enumerate}
\end{example}

In computing, paths arise as execution sequences of automata or
transition systems, and are sometimes called traces. Sets of traces
are are models for the behaviours of concurrent or distributed
computing systems. Lifting along Theorem~\ref{P:lr-pow-lifting}
constructs modal quantales of traces at powerset level; lifting along
Theorem~\ref{P:lr-lifting} the associated modal convolution quantales
of $Q$-weighted traces. Once again, as finite traces are considered,
semirings suffice as weight algebras. Algebras of weighted paths are
important for quantitative analysis of systems or for the design of
algorithms. The construction of the modal powerset quantale over the
path category of a digraph has recently been extended to
higher-dimensional modal Kleene algebras and higher-dimensional
(poly)graphs~\cite{CalkGMS20}.

\begin{example}[Modal Incidence Algebras over Categories of Segments
  and Intervals]\label{ex:segments-intervals}

  The arrows or pairs in poset categories, as in
  Example~\ref{ex:multimon}(\ref{ex:poset}), represent (closed)
  segments of the poset; those of a linear order represent (closed)
  intervals. By Theorem~\ref{P:lr-pow-lifting}, for any modal quantale
  $Q$, the convolution algebra or category algebra over a poset
  category with values in $Q$ forms a modal quantale.

  Without the modal structure, such convolution quantales have already
  been studied~\cite{DongolHS21}.  When the underlying partial or
  total orders are \emph{locally finite}, that is, all segments and
  intervals are finite (all arrows in the poset category are finitely
  decomposable), values can be taken in (semi)rings and the
  convolution algebras become the incidence algebras \`a la
  Rota~\cite{Rota64}. The general setting supports algebraic
  generalisations of duration calculi~\cite{ZhouH04,DongolHS21}; it
  specialises to interval and interval temporal
  logics~\cite{HS91,Mos12} for powerset liftings. The additional modal
  structure yields richer algebras for reasoning about states as well
  as intervals and supports mixed modalities over weighted intervals
  and their endpoints. This, and in particular applications to
  weighted and probabilistic interval temporal logics and duration
  calculi, remains to be explored.\qed
\end{example}

Our next example considers $\ell r$-multisemigroups that arise from
composing digraphs, yet we specialise to finite partial orders. We
show how locality can be obtained by introducing interfaces.

\begin{example}[Weighted Poset Languages]\label{ex:pomset}
We restrict our attention to finite posets.
\begin{enumerate}
\item Finite posets form partial $\ell r$-multisemigroups (based on
  classes) with respect to serial composition, which is the disjoint
  union of posets with all elements of the first poset preceding that
  of the second one in the order of the composition.  This yields
  partial monoids with the empty poset as unit, and hence partial
  $\ell r$-semigroups in which $\ell$ and $r$ map every poset to the
  empty poset. The algebra is therefore not local and does not form a
  category. The convolution algebras are therefore weakly local
  quantales, but the modal structure is trivial as $\Dom$ and $\Cod$
  map the empty poset to $\bot$ and any other element to
  $\id_{\{1\}}$. The powerset lifting yields poset languages.
\item The points of finite posets can be labelled with letters from
  some alphabet and equivalence classes of such labelled posets can be
  taken with respect to isomorphism that preserves the order structure
  and labels, but forgets the names of nodes. This leads to partial
  words~\cite{Grabowski81}, which are also known as \emph{pomsets} in
  concurrency theory. The serial composition passes to a total
  operation on equivalence classes, which makes the resulting
  $\ell r$-monoid a monoid and hence a category.
  
\item Pomsets can be equipped with interfaces~\cite{Winkowski77}. The
  source interface of a pomset consists of all its minimal elements
  (with their labels); its target interface of all its maximal
  elements (again with their labels). Pomsets with interfaces form a
  partial $\ell r$-semigroup with $\ell$ mapping every pomset to its
  source interface, $r$ mapping every poset to its target interface,
  and composition defined by gluing pomsets on their interfaces
  whenever they match, and extending the order as in the previous
  example.  The partial $\ell r$-semigroup of such pomsets with
  interfaces is local and hence a category. Winkowski also defines a
  parallel composition that turns pomsets with interfaces into a weak
  monoidal category in which parallel composition is a partial
  tensor. Details are beyond the scope of this article, but see
  Section~\ref{S:modal-chantales} for initial steps that extend modal
  quantales towards concurrency. \qed
\end{enumerate}
\end{example}

Winkowski's pomsets with interfaces have recently been extended to
posets in which interfaces may be arbitrary subsets of the sets of minimal or
maximal elements of posets~\cite{FahrenbergJST20}. Our lifting results
extend to these.  Compositions of digraphs with interfaces can also be
defined differently, simply by juxtaposing these objects and then
making interface nodes disappear in the composition. This yields again
categories, yet with different units given by identity
relations~\cite{Hotz65,FioreD13}.  The usual liftings therefore apply.
In all these examples one considers an additional operation of
parallel composition, which is simply disjoint union of posets or
dags. For these, the powerset lifting yields standard models of
concurrency.

The following example lifts a local $\ell r$-multimagma to a modal
prequantale.

\begin{example}[Weighted Path Languages in Topology]\label{ex:continuous-paths}~
\begin{enumerate}
\item A \emph{path} in a set $X$ is a map $f:[0,1]\to X$, where it is
  usually assumed that $X$ is a topological space and $f$
  continuous. The source of path $f$ is $\ell(f)=f(0)$; its target
  $r(f)=f(1)$.  Two paths $f$ and $g$ in $X$ can be composed whenever
  $r(f)=f (1) = g (0)=\ell(g)$, and then
\begin{equation*}
  (f\cdot g) (x) =
  \begin{cases}
    f (2x) & \text{ if } 0\le x\le \frac{1}{2},\\
g(2x-1) & \text{ if } \frac{1}{2}\le x\le 1.
  \end{cases}
\end{equation*}
The parameterisation destroys associativity of composition, hence
$(X^{[0,1]},\cdot,\ell,r)$ is a local partial $\ell r$-magma.  The
powerset lifting to $\Pow (X^{[0,1]})$ satisfies the properties of
Lemma~\ref{P:lr-magma-lift}, but even weak locality fails due to the
absence of associativity in $X^{[0,1]}$ and, accordingly,
$\Pow (X^{[0,1]})$. The same failure occurs when lifting to a
convolution quantale $Q^{X^{[0,1]}}$ along Theorem~\ref{P:lr-lifting}:
associativity of the modal (pre)quantale $Q$ of weights makes no
difference; the convolution algebra forms merely a modal prequantale
whenever $Q$ does.

\item Path composition is of course associative up-to homotopy. The
  associated local partial $\ell r$-semigroup can then be lifted like
  any other category.  Associativity can be enforced in a more
  fine-grained way when paths are considered up-to reparametrisation
  equivalence \cite{journals/jhrs/FahrenbergR07}: two paths
  $f,g:[0,1]\to X$ are \emph{reparametrisation equivalent} if there is
  a path $h$ and surjective increasing maps
  $\varphi,\psi:[0,1]\to[0,1]$ such that $f=h\circ \varphi$ and
  $g=h\circ\psi$.  Composition of reparametrisation equivalence
  classes of paths (also called \emph{traces}) is associative and the
  convolution algebra a modal quantale.

\item Alternatively, categories of \emph{Moore paths} can be defined
  on intervals of arbitrary length~\cite{Brown06}.  A path is then a
  (continuous) map $f:[0,n]\to X$ and, writing $|f|$ instead of $n$
  and likewise,
\begin{equation*}
  (f\cdot g) (x) =
  \begin{cases}
    f (x) & \text{ if } 0\le x\le |f|,\\
g (x-|f|) & \text{ if } |f|\le x\le |f|+|g|.
  \end{cases}
\end{equation*}
Lifting along Theorem~\ref{P:lr-pow-lifting} now yields a modal
convolution quantale appropriate, for instance, for interval temporal
logics~\cite{Mos12} or durational calculi~\cite{ZhouH04} for hybrid or
continuous dynamical systems (paths could correspond to trajectories
for initial value problems for vector fields or systems of
differential equations). Lifting along Theorem~\ref{P:lr-lifting}
yields algebras that may be useful for modelling for weighted,
probabilistic or stochastic systems with continuous dynamics. \qed
\end{enumerate}
\end{example}

\def\Int{\mathrm{Int}}
\begin{example}[$\Delta$-sets]
  \emph{A presimplicial set} \cite{RourkeSanderson71} $K$ is a
  sequence of sets $(K_n)_{n\geq 0}$ equipped with face maps
  $d_i:K_n\to K_{n-1}$, $i\in \{0,\dotsc,n\}$, satisfying the
  simplicial identities $d_i\circ d_j=d_{j-1}\circ d_i$ for all $i<j$.
  Elements of $K_n$ are $n$-simplices.  Presimplicial sets generalise
  digraphs: $K_0$, $K_1$ are the sets of vertices and edges of a
  digraph, respectively; $d_0,d_1:K_1\to K_0$ are the source and
  target maps.  Higher simplices represent compositions of edges.  For
  edges $x,y\in X_1$, for example, a $2$-simplex $z\in K_2$ with
  $d_2(z)=x$ and $d_0(z)=y$ reflects the fact that the edge $d_1(z)$
  is a composition of $x$ and $y$.  For given $x$ and $y$ there may be
  many such $2$-simplices $z$, or none at all.
  
  There are at least two ways of constructing an
  $\ell r$-multisemigroup from a presimplicial set $K$.

  First, for $0\leq i\leq n$ and $x\in K_n$, denote
  \begin{equation*}
    s_i(x)=(d_{i+1}\circ d_{i+2}\circ \dots \circ
    d_{n})(x)\qquad\text{ and }\qquad
    t_i(x)=(d_{0}\circ d_{1}\circ \dots \circ d_{n-i-1})(x).
  \end{equation*}
  By definition, $s_i(x),t_i(x)\in K_i$.  The set of simplices
  $K=\bigsqcup_{n\geq 0} K_n$ forms a graded $\ell r$-multisemigroup
  $(K, \mcomp,\ell,r)$ with
    \begin{equation*}
      x\in y\mcomp z\ \Leftrightarrow\ \exists i.\  y=s_i(x) \land z=t_{n-i}(x)
    \end{equation*}
    and $\ell(x)=s_0(x)$, $r(x)=t_0(x)$.  Associativity follows from
    the relationship $t_j(s_i(x))=s_j ( t_k(x))$ that holds for all
    $x\in X_n$ and $0\leq i,j,k\leq n$ such that $i+k=n+j$.

    In general, the graded $\ell r$-multisemigroup $(K,\mcomp)$ is
    neither local nor partial.  Locality and partiality hold if $K$ is
    the nerve of a category $C$ (we omit degeneracies).  In this case,
    elements of $K_n$ are sequences of morphisms
    \begin{equation*}
    	x=(c_0 \xrightarrow{\alpha_1} c_1\xrightarrow{\alpha_2} \dotsm \xrightarrow{\alpha_n} c_n)
      \end{equation*}
    with $\ell(x)=(c_0)$, $r(x)=(c_n)$, while $\mcomp$ is the
    concatenation of sequences.
    
    The second construction uses \emph{simplices with interfaces} in
    $K$, which are triples $(s_i(x),x,t_j(x))$ for $x\in K_n$,
    $0\leq i,j\leq n$.  The set $\Int(K)$ of all simplices with
    interfaces in $K$ forms an $\ell r$-multisemigroup with
	    \begin{multline*}
      (s_p(x),x,t_q(x))\in (s_i(y),y, t_j(y))\mcomp (s_k(z),z,t_l(z))
      \\
      \Leftrightarrow
       s_p(x)=s_i(y) \land
       t_j(y)=s_k(z) \land
       t_q(x)=t_l(z)\land
       y=s_{u}(x) \land z=t_{n-u+j}(x),
    \end{multline*}
    for $x\in K_n$, $y\in K_u$, $z\in K_{n-u+j}$, and
    \begin{align*}
    	\ell(s_i(x),x,t_j(x))&=(s_i(s_i(x)),s_i(x),t_i(s_i(x)))=(s_i(x),s_i(x),s_i(x)),\\
    	r(s_i(x),x,t_j(x))&=(s_j(t_j(x)),t_j(x),t_j(t_j(x))=(t_j(x),t_j(x),t_j(x)).
    \end{align*}
	There is an obvious embedding
    \begin{equation*}
      K\ni x \mapsto (s_0(x),x,t_0(x)) \in \Int(K)
    \end{equation*}
    of $\ell r$-multisemigroups; hence, $\Int(K)$ is, again, neither
    partial nor local.  As with the first construction, partiality and
    locality hold for nerves of categories.  An element $x$ of
    $\Int(K)$ is then a sequence of composable morphisms of $C$ with
    distinguished initial segment $\ell(x)$ and final segment $r(x)$,
    while $x\mcomp y$ is a concatenation of $x$ and $y$ with $r(x)$
    and $\ell(y)$ identified.\qed
\end{example}

\begin{example}[Precubical sets]
  \emph{A precubical set} $X$ \cite{thesis/Serre51,book/Grandis09} is
  a sequence of sets $(X_n)_{n\geq 0}$ equipped with face maps
  $d^\varepsilon_{i}:X_n\to X_{n-1}$, $1\leq i\leq n$,
  $\varepsilon\in\{0,1\}$, satisfying the identities
  $d^\varepsilon_i\circ d^\eta_j = d^\eta_{j-1}\circ d^\varepsilon_i$
  for $i<j$ and $\varepsilon,\eta\in\{0,1\}$.  Like presimplicial
  sets, precubical sets generalise digraphs: $X_0$ and $X_1$ may be
  regarded as  sets of vertices and edges, respectively, and
  $d^0_1,d^1_1:X_1\to X_0$ as source and target maps.  Higher cells
  represent equivalences between paths. A square $x\in X_2$, for
  example, reflects the fact that the paths $(d^0_1(x),d^1_2(x))$ and
  $(d^0_1(x),d^1_1(x))$ are equivalent.

  For a subset $A=\{a_1<\dotsm<a_k\}\subseteq [n]$ and
  $\varepsilon\in\{0,1\}$ define the iterated face map
  $d^\varepsilon_A:X_n\to X_{n-|A|}$ by
\begin{equation*}
	d^\varepsilon_A(x)=d^\varepsilon_{a_1}\circ \dotsm \circ d^\varepsilon_{a_k}(x).
\end{equation*}
The precubical set $X$ then forms an $\ell r$-semigroup
$(X,\mcomp,\ell,r)$ where
	\begin{equation*}
          x\in y\mcomp z \Leftrightarrow \exists A\subseteq[n].\  y=d^0_A(x) \land z=d^1_{[n]\setminus A}(x)
	\end{equation*}
	and $\ell(x)=d^0_{[n]}(x)\in X_0$, $r(x)=d^1_{[n]}(x)\in X_0$
        for all $x\in X_n$.  Like in the previous example, the
        $\ell r$-multisemigroup $X$ is neither partial nor local.
	
        A special case of this example is the shuffle multimonoid from
        Example \ref{ex:multimon}.  Let $\Sigma$ be a finite alphabet,
        $X_n$ the set of all words of length $n$, and
        $d^\varepsilon_i:X_n\to X_{n-1}$ the map that removes the
        $i$--th letter.  Then $X=(X_n,d^\varepsilon_i)$ is a
        precubical set and the associated $\ell r$-multisemigroup
        $(X,\mcomp,\ell,r)$ is the shuffle multimonoid on
        $\Sigma$.\qed
 \end{example}

Next we present an example of a weakly local modal convolution
quantale.

\begin{example}[Weighted Assertions in Separation Logic]
  Revisiting the non-local partial $\ell r$-semigroup of heaplets,
  lifting along Theorem~\ref{P:lr-lifting} yields a weakly local modal
  quantale as the convolution algebra, yet once again with trivial
  domain/codomain structure, as the empty heaplet is the only
  unit~\cite{DongolHS16}. The modal structure is trivial as there are
  no elements between $\emptyset$ and $\{\varepsilon\}$.  This models
  weighted assertions of separation logic, including fuzzy or
  probabilistic ones.\qed
\end{example}

Our final example is a lifting from a proper $\ell r$-multisemigroup
that is not even partial.

\begin{example}[Weighted Shuffle languages]\label{ex:shuffle}
  We have already seen that words under shuffle form a proper
  $\ell r$-multisemigroup that is not partial, but
  local~\cite{CranchDS20}.  There is only one unit---the empty
  word. Liftings to $Q$-weighted shuffle languages are discussed
  in~\cite{CranchDS20a}. The domain/codomain structure of the
  convolution algebra is trivial.\qed
\end{example}


\section{Modal Concurrent Convolution Quantales}\label{S:modal-chantales}

Correspondences for biquantales and relational bimonoids that satisfy
certain weak interchange laws have already been
studied~\cite{CranchDS20a}.  The two operations in this setting are
interpreted as a serial or sequential and a concurrent composition,
yet the latter need not be commutative. We outline a simple modal
extension.

An \emph{interchange $\ell r$-multisemigroup} is formed by two local
$\ell r$-multisemigroups $(X,\odot_i,\ell_i,r_i)_{i\in \{0,1\}}$ that
interact via the \emph{weak interchange law}
\begin{equation*}
  (w \odot_1 x)\odot_0 (y\odot_1 z) \subseteq (w\odot_0 y) \odot_1
  (x\odot_0 z).
\end{equation*}

It is helpful to think of $\odot_0$ as a \emph{horizontal}
multioperation and of $\odot_1$ as a \emph{vertical} one; weak
interchange can then be represented graphically:
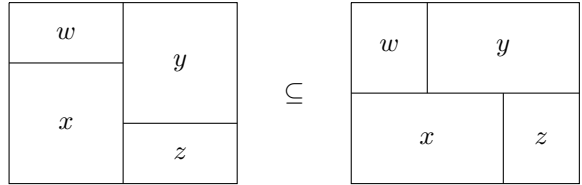
\begin{equation*}
  \begin{tikzpicture}[-, x=.5cm, y=.4cm]
    \begin{scope}
      \draw (0,0) to (6,0) to (6,-6) to (0,-6) to (0,0);
      \draw (3,0) to (3,-6);
      \draw (0,-2) to (3,-2);
      \draw (3,-4) to (6,-4);
      \node at (1.5,-1) {$w$};
      \node at (1.5,-4) {$x$};
      \node at (4.5,-2) {$y$};
      \node at (4.5,-5) {$z$};
      \node at (7.5,-3) {$\subseteq$};
    \end{scope}
    \begin{scope}[shift={(9,0)}]
      \draw (0,0) to (6,0) to (6,-6) to (0,-6) to (0,0);
      \draw (0,-3) to (6,-3);
      \draw (2,0) to (2,-3);
      \draw (4,-3) to (4,-6);
      \node at (1,-1.5) {$\vphantom{y}w$};
      \node at (2,-4.5) {$x$};
      \node at (4,-1.5) {$y$};
      \node at (5,-4.5) {$z$};
    \end{scope}
  \end{tikzpicture}
\end{equation*}

Similarly, an \emph{interchange quantale} is formed by two quantales
$(Q,\le,\cdot_i,1_i)_{i\in\{0,1\}}$ that interact via the weak
interchange law
\begin{equation*}
    (\alpha \cdot_1 \beta)\cdot_0 (\gamma\cdot_1 \delta) \le (\alpha\cdot_0 \gamma) \cdot_1
  (\beta\cdot_0 \delta).
\end{equation*}

In any interchange quantale, substituting $1_0$ and $1_1$ in the
interchange law yields $1_0\le 1_1$.  Moreover, $1_0=1_1=1$ leads to
an Eckmann-Hilton collapse where weak interchange implies the smaller
variants
\begin{alignat*}{5}
  \alpha\cdot_0 \beta &\le \alpha\cdot_1 \beta, & \alpha\cdot_0 \beta &\le\beta \cdot_1 \alpha,\\
\alpha\cdot_0(\beta\cdot_1 \gamma) &\le (\alpha\cdot_0 \beta)\cdot_1 \gamma,&\qquad (\alpha\cdot_1
\beta)\cdot_0 \gamma &\le \alpha\cdot_1 (\beta \cdot_0 \gamma),\\
\alpha\cdot_0(\beta\cdot_1 \gamma) &\le \beta \cdot_0(\alpha\cdot_1 \gamma),& (\alpha\cdot_1
\beta)\cdot_0 \gamma &\le (\alpha \cdot_0 \gamma)\cdot_1 \beta.
\end{alignat*}

Correspondence triangles between such interchange laws in double
$\ell r$-multimagmas $X$, double prequantales $Q$ and double
prequantales $Q^X$ (where no interchange laws are assumed) have
already been established~\cite{CranchDS20a} .

Here we extend interchange $\ell r$-multisemigroups to modal
interchange $\ell r$-multisemigroups by imposing locality for $\ell_i$
and $r_i$.  Interchange quantales are extended to \emph{modal
  interchange quantales} by adding domain and codomain operations
$\dom_i$ and $\cod_i$, for $i\in\{0,1\}$, with the usual axioms.

Combining the results for interchange algebras~\cite{CranchDS20a} with
those in this article then yields the following lifting result.
\begin{theorem}
  Let $X$ be a local interchange $\ell r$-multisemigroup and $Q$ a
  modal interchange quantale. Then $Q^X$ is a modal interchange
  quantale.
\end{theorem}

In addition, the results in this article combine with those for
interchange algebras to structures with only weak locality or even
without locality in a modular way and extend to full correspondence
triangles.  Smaller interchange laws in two of $X$, $Q$ and $Q^X$ also
yield a small interchange law in the remaining
algebra~\cite{CranchDS20a}.

In concrete models such as the pomset and graph categories with
interfaces discussed in Section~\ref{S:examples}, it seems reasonable
to impose additional laws. For pomsets with interfaces, the parallel
composition is simply disjoint union, and interfaces extend
accordingly. The unit of parallel composition is the empty pomset with
empty interfaces, so that $\ell_1$ and $r_1$ become trivial.

Parallel composition $\odot_1$ can then be made commutative in
$X$~\cite{Winkowski77}. This yields commutativity of~$\ast_1$ in $Q^X$
whenever $\cdot_1$ commutes in $Q$, plus the two usual remaining
correspondences~\cite{CranchDS20a}.

Moreover, the serial source and target maps $\ell_0$ and $r_0$ may be
assumed to distribute with parallel composition $\odot_1$ whenever
this composition is defined:
\begin{equation*}
  \ell_0(x\odot_1 y) \subseteq \ell_0(x)\odot_1\ell_0(y)\qquad\text{ and
  }\qquad r_0(x\odot_1 y) \subseteq  r_0(x)\odot_1 r_0(y).
\end{equation*}
This means that $\ell_0$ and $r_0$ are multisemigroup automorphisms
with respect to $(X,\odot_1)$.  They may be extended to
$\ell r$-multisemigroup morphisms on pomsets with interfaces, but we
are not interested in this property in this section.  Alternatively,
equational variants of the two automorphism laws may be assumed. We
impose corresponding laws in $Q$, 
 \begin{equation*}
\dom_0(\alpha\cdot_1 \beta) \le \dom_0(\alpha) \cdot_1 \dom_0(\beta)\qquad\text{ and }\qquad
  \cod_0(\alpha\cdot_1 \beta) \le \cod_0(\alpha) \cdot_1 \cod_0(\beta),
\end{equation*}
and refer to these laws uniformly as \emph{weak automorphism laws} in
the sequel.  We use $[P]_0$ and $[P]_1$ in the proof that follows,
mapping to different units as needed.

\begin{lemma}
  Let $X$ be a local interchange $\ell r$-semigroup and $Q$ a modal
  interchange quantale.  If the weak automorphism laws hold in $X$ and
  $Q$, then they hold in hold in $Q^X$: 
\begin{equation*}
  \Dom_0(f\ast_1 g) \le \Dom_0(f) \ast_1 \Dom_0(g)\qquad\text{ and }\qquad
  \Cod_0(f\ast_1 g) \le \Cod_0(f) \ast_1 \Cod_0(g).
\end{equation*}
\end{lemma}

\begin{proof}
\begin{align*}
  \Dom_0(f\ast_1 g)(x) 
&=\Sup_u\dom_0\left(\Sup_{v,w}f(v)\cdot_1 g(w))\cdot_1 [u\in v\odot_1 w]_1\right)\cdot_0
  \delta_{\ell_0(u)}(x)\\
&=\dom_0\left(\Sup_{v,w}f(v)\cdot_1 g(w)\right)\cdot_0 [x\in \ell_0(v\odot_1 w)]_0\\
&=\Sup_{v,w}\dom_0(f(v)\cdot_1 g(w))\cdot_0 [x\in \ell_0(v\odot_1 w)]_0\\
&\le\Sup_{v,w}\dom_0(f(v))\cdot_1 \dom_0(g(w))\cdot_0 [x\in
  \ell_0(v)\odot_1 \ell_0(w)]_0\\
&=\Sup_{t,u} (\Sup_v\dom_0(f(v))\cdot_1 \delta_{\ell(v)}(t))\cdot_1
  (\Sup_w\dom_0(g(w))\cdot_1 \delta_{\ell(w)}(u)\cdot_1 [x\in t\odot u]_1\\
&= (\Dom_0(f)\ast_1 \Dom_0(g))(x).
\end{align*}
The proof for $r$, $\cod$ and $\Cod$ is dual.
\end{proof}

\begin{remark}
For the remaining correspondences,  observe that
\begin{align*}
  \Dom_0(\delta^\alpha_x\ast_1 \delta^\beta_y)(z) &= \dom_0(\alpha \cdot_1 \beta) \cdot_0 [z\in 
                                                      \ell(x\odot_1 y)]_0,\\
  (\Dom_0(\delta^\alpha_x)\ast_1 \Dom_0(\delta^\beta_y))(z)
                                                    &= \dom_0(\alpha) \cdot_1 \dom_0(\beta)\cdot_1 [z\in \ell_0(x)\odot_1 \ell_0(y)]_1.
\end{align*}
\begin{itemize}
\item For the correspondence from $X$ and $Q^X$ to $Q$, suppose the
  weak automorphism laws hold in $X$ and $Q^X$. Assume that there exist
  $x,y,z\in X$ such that $z \in \ell_0(x)\odot_1 \ell_0(y)$. Then
\begin{equation*}
  \dom_0(\alpha\cdot_1 \beta) = \Dom_0(\delta^\alpha_x\ast_1 \delta^\beta_y)(z)
  \le (\Dom_0(\delta^\alpha_x)\ast_1 \Dom_0(\delta^\beta_y))(z)=\dom_0(\alpha) \cdot_1 \dom_0(\beta).
\end{equation*}
\item For the correspondence from $Q$ and $Q^X$ to $X$, suppose
  the weak automorphism laws hold in $Q$ and $Q^X$. Assume that there exist
  $\alpha,\beta\in Q$ such that $\dom_0(\alpha\cdot_1\beta) \neq 0$.
  It then follows from the above observation that
  $\Dom_0(\delta^\alpha_x\ast_1\delta^\beta_y)(z) \le
  (\Dom_0((\delta^\alpha_x)\ast_1 \Dom_0(\delta^\beta_y))(z)$ implies
  $\ell_0(x\odot_1 y)\le \ell_0(x)\odot_1 \ell_0(y)$.
\end{itemize}
\end{remark}

Further adjustments to concrete semantics for concurrent systems are
left for future work.


\section{Quantitative Modal Algebras}\label{S:dualities}

The domain and codomain yield modal diamond and box operators on the
convolution algebra, as usual for modal
semirings~\cite{DesharnaisS11}. Many properties developed for such
semirings transfer automatically to quantales, yet additional
properties hold.  Forward and backward modal diamond operators can be
defined in a modal quantale $Q$, as on any modal semiring, for
$\alpha,\beta\in Q$, as
\begin{equation*}
  |\alpha\rangle \beta = \dom(\alpha\beta)\qquad\text{ and }\qquad \langle \alpha| \beta = \cod(\beta\alpha).
\end{equation*}
The two operators are related by opposition and the following
conjugation, as for modal semirings~\cite{DesharnaisS11}.
\begin{lemma}\label{P:dia-conj}
In every modal quantale $Q$,  for all $\alpha\in Q$ and $\rho,\sigma\in Q_\dom$, 
\begin{equation*}
  \rho\cdot |\alpha\rangle \sigma = \bot \Leftrightarrow \langle \alpha|\rho \cdot \sigma =\bot.
\end{equation*}
\end{lemma}
\begin{proof}
 $  \rho|\alpha\rangle \sigma = \bot
\Leftrightarrow \rho\dom(\alpha\sigma) = \bot
\Leftrightarrow  \rho\alpha\sigma = \bot
\Leftrightarrow \cod(\rho\alpha)\sigma = \bot
\Leftrightarrow \langle \alpha|\rho \cdot \sigma = \bot$.
 \end{proof}
 The laws
 $\cod(\alpha)\beta = \bot \Leftrightarrow \alpha\beta = \bot
 \Leftrightarrow \alpha\dom(\beta)=\bot$ used in this proof hold in
 any modal quantale or modal semiring because
 $\dom(\alpha)=\bot\Leftrightarrow \alpha=\bot\Leftrightarrow
 \cod(\alpha)=\bot$ and by locality~\cite{DesharnaisS11}.  We have
 seen similar laws for $\ell r$-multisemigroups in
 Lemma~\ref{P:local-alt} in Section~\ref{S:lr-semigroups}.

 Another consequence of locality is that
 $|\alpha\rangle \beta = |\alpha\rangle \dom(\beta)$, so that $\beta$
 is automatically a domain element. We henceforth indicate this by
 writing $|\alpha\rangle \rho$ for $\rho\in Q_\dom$.  Moreover,
 $|\alpha\beta\rangle = |\alpha\rangle\circ |\beta\rangle$ and
 $\langle \alpha\beta| = \langle \beta|\circ \langle \alpha|$, as
 is standard in modal logic. Without locality, only inequalities hold.

\begin{lemma}\label{P:dia-sup}
  In any modal quantale $Q$, the operators $|\alpha\rangle$ and
  $\langle \alpha|$ preserve arbitrary sups in $Q_\dom$, for all
  $\alpha\in Q$.
\end{lemma}
\begin{proof}
  The $\dom$ and $\cod$ operations preserve all sups in the complete
  lattice $Q_\dom$ (which is equal to
  $Q_\cod$)~\cite{FahrenbergJSZ20}.  Therefore, diamonds preserve all
  sups, too: for all $P\subseteq Q_\dom$,
\begin{equation*}
  |\alpha\rangle \big(\Sup P\big) = \dom\big(\alpha\cdot \big(\Sup P\big)\big)
= \dom\big(\Sup\{\alpha\rho \mid \rho\in P\}\big)
 = \Sup\{\dom(\alpha\rho)\mid \rho\in P\}
= \Sup \{|\alpha\rangle \rho\mid \rho\in P\}.
\end{equation*}
Sup preservation of $\langle \alpha|$ follows by duality.
\end{proof}

\begin{remark}
  In modal semirings, diamond operators are strict and preserve all
  finite sups, that is,
  $|\alpha\rangle \bot = \bot =\langle \alpha|\bot$,
  $|\alpha\rangle (\rho\sup\sigma) = |\alpha\rangle
  \rho\sup|\alpha\rangle \sigma$ and
  $\langle \alpha|(\rho\sup\sigma) = \langle \alpha|\rho\sup\langle
  \alpha|\sigma$ (using quantale notation with $\sup$ in place of $+$
  and $\bot$ in place of $0$).
\end{remark}

An immediate consequence of sup-preservation is that modal quantales
admit box operators.
 \begin{proposition}\label{P:dia-galois}
   In any modal quantale $Q$, the operators $|\alpha\rangle$
   and $\langle \alpha|$ have right adjoints in $Q_\dom$:
\begin{equation*}
  [\alpha|\rho = \Sup\{\sigma\mid |\alpha\rangle \sigma \le \rho\}\qquad\text{ and }\qquad |\alpha]\rho
  = \Sup\{\sigma\mid \langle \alpha|\sigma \le \rho\}. 
\end{equation*}
They satisfy 
\begin{equation*}
  |\alpha\rangle \rho \le \sigma \Leftrightarrow \rho \le [\alpha|\sigma\qquad\text{ and }\qquad
  \langle \alpha|\rho\le \sigma \Leftrightarrow \rho \le |\alpha]\sigma. 
\end{equation*}
 \end{proposition}
\begin{proof}
  The diamonds $|\alpha\rangle$ and $\langle \alpha|$ preserve
  arbitrary sups by Lemma~\ref{P:dia-sup}. Hence they have right
  adjoints (by the adjoint functor theorem), defined as above.
\end{proof}
It follows that $|\alpha\beta] = |\alpha]\circ |\beta]$, because,
for all $\sigma\in Q_\dom$, 
\begin{equation*}
  \sigma \le |\alpha\beta]\rho \Leftrightarrow
  \langle\alpha\beta|\sigma\le \rho \Leftrightarrow
  \langle\beta|\langle \alpha|\sigma \le \rho 
\Leftrightarrow \langle \alpha|\sigma \le |\beta]\rho
\Leftrightarrow \sigma
  \le |\alpha]|\beta]\rho,
\end{equation*}
and dually $[\alpha\beta| = [\beta|\circ [\alpha|$. 

The following lemma relates boxes and diamonds with laws that do not
mention modalities.
\begin{lemma}\label{P:dia-demod}
In every modal quantale, 
\begin{alignat*}{5}
  |\alpha\rangle \rho \leq \sigma &\Leftrightarrow \alpha\rho \le
  \sigma\alpha, &\qquad 
  \langle \alpha|\rho\leq \sigma &\Leftrightarrow
  \rho\alpha\le \alpha\sigma,\\
  \rho\le |\alpha]\sigma &\Leftrightarrow \rho\alpha\le \alpha\sigma, &
\rho\le [\alpha|\sigma &\Leftrightarrow \alpha\rho\le \sigma\alpha.
  \end{alignat*}
\end{lemma}
\begin{proof}
  We consider only the first law. The others follow from duality and
  the adjunctions. First, suppose $|\alpha\rangle \rho \le \sigma$,
  that is, $\dom(\alpha\rho)\le \sigma$. Then
  $\alpha\rho=\dom(\alpha\rho)\alpha\rho \le \sigma\alpha\rho \le
  \sigma\alpha$. For the converse direction, suppose
  $\alpha\rho\le \sigma\alpha$. Then
  $\dom(\alpha\rho) \le \dom(\sigma\alpha) = \sigma \dom(\alpha) \le
  \sigma$.
\end{proof}
Therefore,
$|\alpha]\sigma = \Sup\{\rho\mid \rho\alpha\le \alpha\sigma\}$ and
$[\alpha|\sigma =\Sup\{\rho\mid \alpha\rho\le \sigma\alpha\}$, and the
diamond operators, as left adjoints, satisfy
$|\alpha\rangle \rho = \Inf\{\sigma\mid \alpha\rho\le \sigma\alpha\}$
and
$\langle \alpha|\rho=\Inf\{\sigma\mid \rho\alpha\le \alpha\sigma\}$.

In any boolean modal semiring and therefore any boolean modal
quantale, the following De Morgan duality relates boxes and
diamonds~\cite{DesharnaisS11}, where $\rho'$ is the complement of
$\rho$ in $Q_\dom$ as before.
\begin{lemma}
  If $Q$ is a boolean modal quantale, then
  $|\alpha]\rho = (|\alpha\rangle \rho')'$ and
  $[\alpha|\rho=(\langle \alpha|\rho')'$.
\end{lemma}
\begin{proof}
  $|\alpha\rangle \rho \le \sigma \Leftrightarrow |\alpha\rangle
  \rho\cdot \sigma' = \bot \Leftrightarrow \rho\cdot \langle \alpha|
  \sigma' = \bot \Leftrightarrow \rho\le (\langle \alpha| \sigma')'$,
  hence $(\langle \alpha| \sigma')'=[\alpha|\rho$. The proof for
  $|\alpha]$ follows by opposition.
\end{proof}
In boolean modal semirings and quantales, we therefore obtain
conjugation laws for boxes as for diamonds.
\begin{lemma}
  In every boolean modal quantale,
\begin{equation*}
  \rho \sup |\alpha] \sigma = 1 \Leftrightarrow [\alpha|\sigma \sup \rho = 1.
\end{equation*}
\end{lemma}
\begin{proof}
 Straightforward from De Morgan duality.
 \end{proof}

 The results for boolean quantales are summarised in
 Figure~\ref{fig:modal-symmetries}. The green and orange
 correspondences hold in any modal quantale, the black ones only in
 the boolean case.  While the De Morgan dualities do not depend on locality,
the conjugations and Galois connections depend on these laws.

 \begin{remark}
   When the $\ell r$-multisemigroup $X$ is finitely decomposable, one
   can thus use boolean modal semirings, that is, boolean monoids
   equipped with the usual domain axioms, in the lifting. Boolean
   monoids are essentially boolean quantales in which only finite sups
   and infs are assumed to exists and required to be preserved,
   including empty ones. It follows that, if $X$ is a finitely
   decomposable local $\ell r$-multisemigroup and $Q$ a boolean modal
   semiring, then $Q^X$ is a boolean modal semiring.
\end{remark}

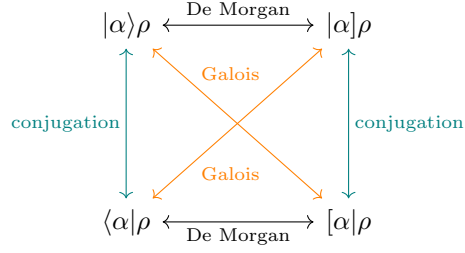
\begin{figure}
  \centering
\begin{tikzcd}[column sep=2cm,row sep=2cm]
  \vert \alpha\rangle \rho \arrow[r, "\text{De Morgan}",
  leftrightarrow]\arrow[d, teal, "\text{conjugation}" ',
  leftrightarrow]\arrow[dr,  orange, "\text{Galois}" near start,
  leftrightarrow]& \vert \alpha] \rho\arrow[d, teal, "\text{conjugation}", leftrightarrow]\\
 \langle \alpha\vert \rho \arrow[r, "\text{De Morgan}" ',
 leftrightarrow]\arrow[ur, orange, "\text{Galois}" ' near start, leftrightarrow] & {[} \alpha\vert \rho
\end{tikzcd}
  \caption{Symmetries between modal operators in boolean modal quantales}
\label{fig:modal-symmetries}
\end{figure}

Next, we consider the modal operators in the convolution quantales
$Q^X$ and relate them with similar properties in $X$ and $Q$. For the
diamonds, 
\begin{align*}
(|f\rangle g)(x) = \Sup_{x \in \ell(y\odot z)} |f(y)\rangle
                         g(z)\qquad\text{ and }\qquad
(\langle f| g)(x) = \Sup_{x\in r(z\odot y)} \langle f(y)| g(z),
\end{align*}
and as upper adjoints are lower adjoints in the dual lattice,
\begin{align*}
(|f] g)(x) = \Inf_{x \in \ell(y\odot z)} |f(y)]
                         g(z)\qquad\text{ and }\qquad
([f| g)(x) = \Inf_{x\in r(z\odot y)} [f(y)| g(z).
\end{align*}
Note that by locality $\ell(x\mcomp y) = \ell(x\mcomp \ell(y))$ and
likewise for $r$, which makes these identities slightly longer, but
more symmetric. Alternatively, we can write
\begin{align*}
  |f\rangle g = \Sup_{x,y,z\in X} |f(y)\rangle g(z) \cdot [x\in \ell
    (y\odot z)] \cdot \delta_x,
  \end{align*}
  and the shape of the laws for the remaining boxes and diamonds is
  then obvious.

Using these definitions we consider the modal operators in relation
and path quantales in more detail. 

\begin{example}[Modal Operators in Powerset Quantales]~
  \begin{enumerate}
  \item In the modal powerset quantale over the path category from
    Example~\ref{ex:traces}(\ref{ex:path-category}), elements are
    sets of paths. Domain elements are formed by the source elements
    in the set and codomain elements by its target elements. Both
    kinds of elements are vertices as paths of length one.  Therefore,
    for any set of paths $A$ and set of vertices $P\subseteq V$ (as
    paths of length one),
    \begin{equation*}
      |A\rangle P=\{ v\in V\mid \exists \pi\in A.\  v\circ \pi\in P\}
    \end{equation*}
    is the set of all vertices in $V$ from which one may reach a
    vertex in $P$ at the end of some path in $A$. Dually,
    $\langle A|P$ models the set of all vertices one may reach from a
    vertex in $P$ at the end of some path in $A$. Moreover,
    \begin{equation*}
      |A] P=\{ v\in V\mid \forall \pi\in A.\ v\circ \pi\in P\}
    \end{equation*}
    models the set of vertices from which one must reach a vertex in
    $P$ at the end of all paths in $A$, and $[A|P$ models the set of
    vertices one must reach from a vertex in $P$ at the end of all
    paths in $A$.
    The shape of the modal operators in algebras of continuous paths
    in Example~\ref{ex:continuous-paths} is very similar.
  \item The modal operators in the modal powerset quantale of the pair
    groupoid---the quantale of binary relations---are the standard
    modal operators with respect to Kripke semantics. Up to
    isomorphism between predicates and subidentity relations, for any
    binary relation $R$ and predicate $P$, we have
    \begin{alignat*}{5}
    |R\rangle P &=\{a\mid \exists b. (a,b)\in R \land  P(b)\}, &\qquad
      \langle R|P &=\{b\mid \exists a.\ (a,b)\in R \land P(a)\},\\
      |R]P &=\{a\mid \forall b. \ (a,b)\in R \Rightarrow  P(b) \},&
      [ R|P &=\{b\mid \forall a. \  (a,b)\in R \Rightarrow  P(a) \}.
    \end{alignat*}
    \qed
 \end{enumerate}
\end{example}

Similar constructions apply to the other powerset algebras over
$\ell r$-multisemigroups outlined in Section~\ref{S:examples}.
Working out details is routine.

\begin{example}\label{ex:weighted-modal-relations}~
  \begin{enumerate}
  \item For weighted relations, the forward diamond is spelled out
    as
\begin{equation*}
(|f\rangle \Sup_c\delta^{g(c,c)}_{(c,c)})((a,b)) =
\Sup_c|f(a,c)\rangle g(c,c)\cdot \delta_a(b).
\end{equation*}
For finite relations, that is, in the finitely decomposable case, we
obtain matrices. In the $2\times 2$ case, for instance,
\begin{align*}
\left\vert
  \begin{pmatrix}
    \alpha_{11} & \alpha_{12}\\
\alpha_{21} & \alpha_{22} 
  \end{pmatrix}
\right\rangle
\begin{pmatrix}
  \beta_{11} & \bot\\
\bot & \beta_{22}
\end{pmatrix}
&=\Dom \left(
  \begin{pmatrix}
    \alpha_{11}\beta_{11} & \alpha_{12} \beta_{22}\\
\alpha_{21}\beta_{11} & \alpha_{22}\beta_{22}
  \end{pmatrix}
\right)\\
&=
  \begin{pmatrix}
    |\alpha_{11}\rangle\beta_{11} \sup |\alpha_{12} \rangle\beta_{22}
    &\bot\\
\bot & |\alpha_{21}\rangle\beta_{11} \sup |\alpha_{22}\rangle\beta_{22}
  \end{pmatrix}.
\end{align*}
The other modal operators satisfy similar laws. The diagonal matrices
are isomorphic to vectors and the computation performed in the above
example corresponds to a linear transformation of a vector, except for
the decoration with diamonds on the diagonal.

\item In the modal quantale of weighted paths,
\begin{equation*}
(|f\rangle \Sup_{v\in V}\delta^{g(v)}_{v})(x)=
\Sup_v|f(v)\rangle g(r(v))\cdot \delta_{\ell(v)}(x).
\end{equation*}
For finite graphs, the element $\delta^{\dom(g(v))}_v$ can again be
seen as a diagonal matrix or vector of weighted vertices, whereas
$f(x)$ can be seen as a matrix labelling pairs of vertices with the
set of weights of the (hom-set) of paths between them. The
multiplication with $\delta^{\dom(g(v))}_v$ then projects on those
paths with targets in $v$ and taking the diamond computes the supremum
of the weight of all paths which end up in a node in $v$. \qed
  \end{enumerate}
\end{example}

From a linear algebra point of view, the use of $\dom$ or $\cod$ in
the computation of weights is irritating. But the results of
Sections~\ref{S:conv-alg} and \ref{S:correspondences} show that the
condition that $Q$ be a modal quantale cannot be weakened. Otherwise,
our correspondence results break.  Yet the approach is
too restricted, for instance,  for deriving probabilistic predicate
transformer semantics in which states and relations carry more general
weights. In the probabilistic quantale from
Example~\ref{ex:value-quantales}, in particular, it is easy to check
that $Q_\ell = \{0,1\}$, which precludes any assignment of non-trivial
weights to states of a system.

Predicate transformers, however, are more general than modal Kleene
algebras: they are simply order-, sup- or inf-preserving functions
between (complete) lattices. The absorption laws
$\dom(\alpha)\cdot \alpha = \alpha$ and
$\alpha\cdot \cod(\alpha)=\alpha$ of domain and codomain, in
particular, seem irrelevant in this setting. Only an action on
weighted relations is needed beyond properties like sup-preservation.
Yet this is given by domain and codomain locality, as discussed at the
beginning of this section. Locality, however, lifts from $X$ to $Q^X$
even if $Q$ is only a quantale.

To explain this, we first redefine domain as
 \begin{equation*}
    \Dom(f)(x) = \Sup_{y\in X} f(y)\cdot \delta_{\ell(y)}(x),
  \end{equation*}
  using $f(y)$ instead of $\dom(f(y))$, and likewise for $\Cod$. Then,
  for locality in $Q^X$ and $Q$ merely a quantale,
\begin{align*}
(\Dom (f \ast g))(v) 
&= \Sup_{x,y,z}f(y)\cdot g(z) \cdot [x\in y\odot z]
  \cdot \delta_{\ell(x)}(v)\\
 &= \Sup_{y,z} f(y)\cdot g(z)\cdot
  [v\in \ell(y\odot z)]\\
 &=\Sup_{y,z}f(y)\cdot g(z)\cdot
   [v\in \ell(y\odot \ell(z))]\\
 &=\Sup_{x,y,z}f(y)\cdot \left(\Sup_wg(w)\cdot
   \delta_{\ell(w)}(z)\right) \cdot [x\in y\odot z]\cdot \delta_{\ell(x)}(v)\\
 &= \Dom (f \ast \Dom (g))(v),
\end{align*}
replaying the proof of Theorem~\ref{P:lr-lifting}. The proof for
$\Cod$ is dual, as usual. Domain export lifts
in the same way:
\begin{align*}
  (\Dom (\Dom (f) \ast g))(v)
&= \Sup_{x,y,z}\left(\Sup_w f(w)\cdot \delta_{\ell(w)}(y)\right)\cdot g(z)\cdot
  [x\in y\odot z]\cdot \delta_{\ell(x)}(v)\\
&=\Sup_{z,w} f(w)\cdot g(z)\cdot
  [v\in \ell(\ell(w)\odot z)]\\
&=\Sup_{z,w} f(w)\cdot g(z)\cdot 
  [v\in\ell(w)\odot \ell(z)]\\
&=(\Dom(f)\ast \Dom(g))(v),
\end{align*}
and once again the proof for codomain export, with $Q$ being merely a
quantale, is dual.  Yet it is easy to check that absorption laws
cannot be lifted this way.

\begin{example}
  Defining forward and backward diamonds as before in the quantale of
  weighted relations, yet using the revised domain definition, yields, 
\begin{equation*}
(|f\rangle \Sup_c\delta^{\dom(g(c,c))}_{(c,c)})((a,b)) =
\Sup_c f(a,c) \cdot g(c,c)\cdot \delta_a(b).
\end{equation*}
And the shape of the backward diamond is then obvious by duality.  In
the matrix case, our previous example now reduces to
\begin{align*}
\left\vert
  \begin{pmatrix}
    \alpha_{11} & \alpha_{12}\\
\alpha_{21} & \alpha_{22} 
  \end{pmatrix}
\right\rangle
\begin{pmatrix}
  \beta_{11} & \bot\\
\bot & \beta_{22}
\end{pmatrix}
&=
  \begin{pmatrix}
    \alpha_{11}\beta_{11} \sup \alpha_{12} \beta_{22}
    &\bot\\
\bot & \alpha_{21}\beta_{11} \sup \alpha_{22}\beta_{22}
  \end{pmatrix}.
\end{align*}
This has the right shape for stochastic matrices and probabilistic
predicate transformers, using the probabilistic quantale from
Example~\ref{ex:value-quantales}. \qed
\end{example}

We have checked with Isabelle that, in the absence of the absorption
laws, domain elements still form a subalgebra and that commutativity
lifts whenever $Q$ is abelian.\footnote{The proofs are routine. We
  have not added them to the Isabelle repository, as theories from
  another repository (Archive of Formal Proofs) need to be downloaded
  by the reader to make them compile.} Idempotency of domain elements
does not lift, which is consistent with the fact that multiplication
of diagonal matrices is not in general idempotent. This is appropriate
for vector spaces and similar structures.  Developing the ``modal''
algebras of such approaches and exploring applications is left for
future work.


\section{Conclusion}\label{S:conclusion}

We have defined $\ell r$-multisemigroups and used them to obtain a
triangle of equational correspondences between modal value quantales
and quantale-valued modal convolution quantales, explaining in
particular the role of locality in axiomatisations of domain semirings
and quantales in light of the typical composition pattern for
categories. Our results yield a generic construction recipe for modal
quantales, in that one gets such quantales for free as soon as the
much simpler underlying $\ell r$-multisemigroup has been
identified. The relevance of this approach is illustrated by many
examples from mathematics and computing.

We have sketched two directions for further work as well: First, in
combination with previous results for concurrent relational monoids
and concurrent quantales~\cite{CranchDS20a}, it seems interesting to
build models for non-interleaving concurrent systems based on
po(m)sets and graphs with interfaces and lift them to modal concurrent
semirings and quantales. Second, the relevance of the weakened domain
and codomain quantales of Section~\ref{S:dualities} for the
quantitative verification of computing systems with weights or
probabilities remains to be explored.

We also aim at a categorification of the approach in terms of Day
convolution. While this yields additional generality, we have chosen a
simpler algebraic approach in this article with a view on verification
applications with proof assistants like Isabelle, where reasoning with
monoidal categories, coends or profunctors may become
unwieldy. Relational monoids, $\ell r$-semirings and the construction
of convolution quantales have already been formalised with
Isabelle~\cite{DongolGHS17}; Isabelle verification components for
separation logic have used a somewhat less general approach of partial
abelian monoids and generalised effect algebras~\cite{DongolGS15}.

Finally, it seems interesting to generalise J\'onsson-Tarski duality
to our weighted setting and study in particular the role played by the
weight algebra $Q$ in this setting.


\bibliographystyle{alpha}
\bibliography{chantale}


\newpage

\appendix

\section{Glossary of Algebraic Structures}\label{A:one}

\paragraph{Multisemigroups and $\ell r$-Semigroups}

\begin{itemize}
\item A \emph{multimagma} $(X,\mcomp)$ is a set $X$ equipped with a
  multioperation $\mcomp:X\times X\to \Pow X$.  
\item A \emph{multisemigroup} is an associative multimagma: for all
  $x,y,z\in X$,
\begin{equation*}
  \bigcup \{x \mcomp v\mid v \in y \mcomp z\} =
  \bigcup \{v \mcomp z\mid v \in x \mcomp y\}.
\end{equation*}
\item A \emph{multimonoid} $(X,\mcomp,E)$ is a multisemigroup with a
  set $E\subseteq X$ of units such that, for all $x\in X$,
\begin{equation*}
  \bigcup \{e\mcomp x\mid e\in E\} = \{x\} = \bigcup\{x\mcomp e\mid
  e\in E\}. 
\end{equation*}
\item A multimagma $X$ is a \emph{partial magma} if $|x\mcomp y| \le
  1$ and a \emph{magma} if $|x\mcomp y|=1$ for all $x,y\in X$.
  These definitions  extend to multisemigroups and multimonoids.
\item A multimagma, multisemigroup or multimonoid $X$ is \emph{local}
  if $u\in x\mcomp y \land y\mcomp z\neq \emptyset \Rightarrow u\mcomp
  z\neq \emptyset$. 
\item An $\ell r$-\emph{multimagma} $(X,\mcomp,\ell,r)$ is a
  multimagma $(X,\mcomp)$ equipped with two operations $\ell,r:X\to X$
  such that, for all $x,y\in X$, 
\begin{equation*}
  x\mcomp y\neq \emptyset \Rightarrow r(x) =\ell(y),\qquad
  \ell(x)\mcomp x = \{x\},\qquad
  x\mcomp r(x) = \{x\}.
\end{equation*}
The definition extends to $\ell r$-multisemigroups, partial and total
algebras.  
\item An $\ell r$-semigroup is \emph{local} if
\begin{equation*}
  \ell (x\mcomp \ell(y)) = \ell(x\mcomp y)\qquad\text{ and }\qquad
  r(r(x)\mcomp y) = r(x\mcomp y).
\end{equation*}
\end{itemize}

\paragraph{Function Systems and Modal Semigroups}

\begin{itemize}
\item  A \emph{function system}~\cite{SchweizerS67} is a structure
$(X,\sscirc,L,R)$ such that $(X,\sscirc)$ is a semigroup and the following
axioms hold (using Schweizer and Sklar's notation):
\begin{xalignat}{3}
L(R(x)) &= R(x),&R(L(x) &= L(x),\tag{2a}\label{eq:2a}\\
L(x)\sscirc x &= x,&x \sscirc R(x) &= x,\tag{2b}\label{eq:2b}\\
L(x\sscirc L(y)) &= L(x\sscirc y),& R(R(x)\sscirc y) &= R(x\sscirc y),\tag{3a}\label{eq:3a}\\
L(x)\sscirc R(y) & = R(y)\sscirc L(x),&&\tag{3b}\label{eq:3b}\\
y\sscirc R(x\sscirc y) & = R(x)\sscirc y.&&\tag{3c}\label{eq:3c}
  \end{xalignat}
  In addition, Schweizer and Sklar consider the identity
\begin{equation}
L(x\sscirc y) \sscirc x = x\sscirc L(y),\tag{D3c} \label{eq:D3c} 
\end{equation}
which is valid in some function systems, but not in all. They point
out that all axioms except (\ref{eq:3c}), which has subsequently been
called \emph{twisted axiom}, hold of binary relations. Schweizer and
Sklar's axioms use function composition $\circ$ instead of $\sscirc$.

\item A \emph{domain semigroup}~\cite{DesharnaisJS09} is a semigroup
  $(X,\cdot)$ with a binary operation $\dom:X\to X$ that satisfies
\begin{align*}
\dom(x)\cdot x &= x, \tag{D1}\label{eq:D1}\\
\dom(x\cdot y) &= \dom(x\cdot \dom(y)),\tag{D2}\label{eq:D2}\\
\dom(\dom(x)\cdot y) &= \dom(x)\cdot \dom(y),\tag{D3}\label{eq:D3}\\
\dom(x)\cdot \dom(y) &= \dom(y)\cdot \dom(x). \tag{D4}\label{eq:D4}
\end{align*}
\item A \emph{modal semigroup} is a domain semigroup $X$ equipped with
  a codomain operation $\cod:X\to X$ that satisfies opposite laws, in
  which the arguments of composition have been swapped, and the
  compatibility laws $\dom\circ \cod = \cod$ and
  $\cod\circ \dom= \dom$.
\end{itemize}
 The axioms (\ref{eq:D3}) and (\ref{eq:D4}) as well as the dual
  codomain axioms are derivable in function systems. Conversely, the
  function system axiom (\ref{eq:3b}) is derivable in modal
  semigroups. Disregarding the twisted law, function systems are
  therefore a more compact, but equivalent axiomatisation of domain
  semigroups.

\paragraph{Dioids and Modal Semirings}
\begin{itemize}
\item A \emph{dioid} is an additively idempotent semiring
$(S,+,\cdot,0,1)$, that is, a semiring $S$ in which $x+x=x$ for all
$x\in S$. As $(S,+,0)$ is a semilattice, the relation
$x\le y\Leftrightarrow x+y=y$ is a partial order with least element
$0$ and in which $+$ and $\cdot$ preserve the order in both arguments.
\item A \emph{modal semiring} is a semiring $S$ equipped with operations
  $\dom,\cod:S\to S$ that satisfy
  \begin{gather*}
    x\le \dom(x)\cdot x,\qquad \dom(x\cdot \dom(y)) = \dom(x\cdot
    y),\qquad \dom(x)\le 1,\qquad \dom(0)=0,\\ \dom(x+y)=\dom(x)+\dom(y),
  \end{gather*}
  as well as opposite axioms for $\cod$ and in which $\dom$ and $\cod$
  are compatible in the sense that $\dom\circ\cod=\cod$ and
  $\cod\circ\dom =\dom$. Adding these axioms to any semiring and
  defining $x\le y\Leftrightarrow x+y=y$ forces additive
  idempotency. Every modal semiring is therefore a dioid.
\end{itemize}

\paragraph{Quantales and Modal Quantales}

\begin{itemize}
\item A \emph{prequantale} is a structure $(Q,\le,\cdot,1)$ such that
  $(Q,\le)$ is a complete lattice, $1$ a unit of composition and
  $\cdot$ preserves all sups in both arguments.
\item A \emph{quantale} is a prequantale $Q$ in which composition is
  associative (hence $(Q,\cdot,1)$ is a monoid).
\item A \emph{modal prequantale} is a prequantale $Q$ with operations
  $\dom,\cod:Q\to Q$ such that the modal semiring axioms hold
  (replacing $0$ with $\bot$ and $+$ with $\sup$), yet the second
  axiom (the locality axiom) for $\dom$ and its opposite for $\cod$
  are replaced by the export axiom
\begin{equation*}
  \dom(\dom(x)\cdot y)) = \dom(x)\cdot \dom(y)
\end{equation*}
and its opposite for $\cod$.
\item A \emph{weakly local modal quantale} is a modal prequantale that
  is also a quantale. 
\item A \emph{modal  quantale} is a quantale $Q$
  equipped with operations $\dom,\cod:Q\to Q$ that satisfy the modal
  semiring axioms (with the obvious syntactic replacements).
\end{itemize}


\section{Multioperations and Ternary Relations}\label{A:two}

  We briefly outline the relationship between multioperations and
  ternary relations.  A relational magma $(X,R)$ is a set $X$ with a
  ternary relation $R$. We write $R^x_{yz}$ instead of $(x,y,z)\in R$
  and $D_{yz}$ whenever $R^x_{yz}$ holds for some $x$.  We have
  previously set up the relationship between object-free categories
  and such relational structures~\cite{CranchDS20}.  The association
  with multioperations is obtained via
\begin{equation*}
  x\in y \mcomp z \Leftrightarrow R^x_{yz}.
\end{equation*}
Accordingly, the ternary relation $R$ is relationally
  associative if
$\forall u,x,y,z.\exists v.\ R^u_{xv}\land R^v_{yz} \Leftrightarrow
\exists v.\ R^v_{xy}\land R^u_{vz}$.
It is local if for all $u,x,y,z$, $R^u_{xy}$ and $D_{yz}$
imply $D_{uz}$. It is weakly functional if $|\{x\mid R^x_{yz}\}|
\le 1$ and functional if $|\{x\mid R^x_{yz}\}| =1$,  for all
$x,y$.  These properties correspond to partiality and totality of
multirelations. Element $e\in X$ is a relational left unit if
$R^x_{ex}$ for some $x$ and $R^x_{ey}$ implies $x = y$ for all $x,y$;
it is a relational right unit if $R^x_{xe}$ for some $x$ and
$R^x_{ye}$ implies $x = y$ for all $x,y$. Definitions of relational
magmas, unitality, relational semigroups and relational monoids are
then analogous to those for multimagmas.  Similar notions have been
studied in~\cite{DongolHS21}.

A relational magma morphism $f:(X,R)\to (Y,S)$ satisfies
$\smash[t]{R^x_{yz}\Rightarrow S^{f(x)}_{f(y)f(z)}}$ for all
$x,y,z\in X$.  It is bounded if $\smash[t]{S^{f(x)}_{uv}}$ implies that there
are $x,y\in X$ such that $R^x_{yz}$, $u=f(y)$ and $v=f(z)$, for all
$x\in X$ and $u,v\in Y$.

\begin{lemma}\label{P:rel-multi-magma}
  Categories of relational magmas and those of multimagmas are
  isomorphic (both for morphisms and for bounded ones).
\end{lemma}

\begin{proof}
Relations of type $X\times Y\times Z$ are in bijective correspondence
with functions of type $Y\times Z \to \Pow X$ via the maps
\begin{equation*}
\rtof (R) = \lambda y, z.\ \{x\mid R^x_{yz}\}\qquad \text{ and }\qquad
\ftor (F) =  \{( x,y,z\mid x \in F(y,z)\}.
\end{equation*}

We extend $\rtof$ and $\ftor$ to functors between categories of
relational magmas and multimagmas by stipulating
$\rtof(f) = f = \ftor(f)$.  
Suppose $f$ is a relational magma morphism. Then $\rtof(f)$ is a
multimagma morphism because
\begin{equation*}
f\left(\rtof(R)(y, z) \right) =\left\{f(x)\middle|  R^x_{yz} \right\}\subseteq
\left\{f(x)\middle| S^{f(x)}_{f(y)f(z)}\right\} =  \rtof(S)(f(y),f(z)).
\end{equation*}
Suppose $f$ is a multimagma morphism. Then $\ftor(f)$ is a relational
magma morphism because
\begin{equation*}
(\ftor(\mcomp_X))^x_{yz} \Leftrightarrow x \in y \mcomp_X y\Rightarrow
f(x) \in f (y)\mcomp_Y f (z) \Leftrightarrow (\ftor(\mcomp_Y))^{f(x)}_{f(y)f(z)}.
\end{equation*}
Moreover, if the relational magma morphism $f$ is bounded, then so is
$\rtof(f)$, because 
\begin{align*}
f(x)\in \rtof(S)(u,v)
&\Leftrightarrow S^{f(x)}_{uv} \\
&\Rightarrow R^x_{yz},\, u=f(y)\text{ and } v=f(z) \text{ for some } y,z\\
& \Leftrightarrow
x\in\rtof(R)(y,z),\,  u=f(y)\text{ and } v=f(z) \text{ for some } y,z.
\end{align*}
Finally, if the multimagma morphism $f$ is bounded, then so is
$\ftor (f)$, because
\begin{align*}
  (\ftor (\mcomp_Y))^{f(x)}_{uv}  
&\Leftrightarrow  f(x) \in u\mcomp_Y v \\
&\Rightarrow x\in y\mcomp_X z,\,  u=f(y)\text{ and } v=f(z) \text{ for
  some } y,z\\ 
&\Leftrightarrow  (\ftor (\mcomp_X))^x_{yz} \land u=f(y)\land v=f(z)
  \text{ for some } y,z.
\end{align*}

It is then easy to check that $\rtof$ and $\ftor$ are mutually inverse
functors. This shows that the categories of relational magmas
with (bounded) morphisms and those of multimagmas with (strong)
morphisms are isomorphic.
\end{proof}

The isomorphisms extend to
relational and multisemigroups, relational monoids and multimonoids,
and their local, partial and total variants. They transfer theorems
between relational structures~\cite{CranchDS20} and the corresponding
multialgebras.

\section{Proofs of Lemmas~\ref{P:lr-magma-lift} and \ref{P:lr-semigroup-lift}}\label{A:three}

\begin{proof}[Proof of Lemma~\ref{P:lr-magma-lift}]
We show proofs up-to opposition.
  \begin{enumerate}
  \item $\ell (r (A)) = \{\ell (r (x)) \mid x\in A\} = \{r (x)\mid
x\in A\} = r (A)$. 
\item 
\begin{align*}
\ell (A) A 
&= \bigcup\{\ell(x)y \mid x,y \in A\text{ and } D_{\ell(x)y}\}\\
&= \bigcup\{\ell(x)y \mid x,y \in A,\,  D_{\ell(x)y}\text{ and } r(\ell(x))=\ell(y)\}\\
&= \bigcup\{\ell(x)y \mid x,y \in A,\,  D_{\ell(x)y}\text{ and }
  \ell(x)=\ell(y)\}\\
&= \bigcup\{\ell(y)y \mid y \in A\}\\
&= \bigcup\{\{y\} \mid y\in A\}\\
&=A. 
\end{align*}
\item 
  $\ell \left(\bigcup \mathcal{A}\right)
= \{\ell (x) \mid x \in \bigcup \mathcal{A}\}
= \{\ell (x) \mid  x \in A \text{ for some }A \in \mathcal{A}\}
= \bigcup\{\ell (A)\mid A \in
  \mathcal{A}\}$.
\item Immediate from (3).
\item We only prove the identity for $\ell (A) r (B)$; the remaining
  ones then follow from (1).
\begin{align*}
  \ell (A) r (B) 
=\bigcup\{\ell(x)r(y) \mid x\in A \text{ and } y\in B\}
=\bigcup\{r(y)\ell(x) \mid x\in A \text{ and } y\in B\}
= r(B)\ell(A).
\end{align*}

\item $\ell (A) = \{\ell (x)\mid x \in A\} \subseteq \{\ell (x)\mid x \in
X\}  = \{x\mid \ell (x) = x\}=E$. 
\item 
  \begin{align*}
    \ell (\ell (A) B)
& =\bigcup\{\ell(\ell(x)y) \mid x\in A,\, y\in B \text{ and }
  D_{\ell(x)y}\}\\
& =\bigcup\{\ell(x)\ell(y) \mid x\in A,\,  y\in B,\,
  D_{\ell(x)y}\text{ and } r(\ell(x)) = \ell(y)\}\\
& =\bigcup\{\ell(x)\ell(y) \mid x\in A,\,y\in B,\, 
  D_{\ell(x)y}\text{ and } \ell(x) = \ell(y)\}\\
& =\bigcup\{\ell(x)\ell(y) \mid x\in A,\,  y\in B \text{ and }
  D_{\ell(y)y}\}\\
& =\bigcup\{\ell(x)\ell(y) \mid x\in A\text{ and } y\in B\}\\
&= \ell(A)\ell(B). 
  \end{align*}
\qedhere
  \end{enumerate}
\end{proof}

\begin{proof}[Proof of Lemma~\ref{P:lr-semigroup-lift}]
  Suppose $X$ is an $\ell r$-multisemigroup. For the first inclusion, it is
  routine to derive (lla) from Lemma~\ref{P:lr-magma-lift}. The claim
  then follows from Proposition~\ref{P:lr-pow} and the adjunction for
  domain. The second inclusion holds by opposition.

Finally, suppose that $X$ is local. Then, writing $r(x)=\ell(y)$ in
place of $D_{xy}$ owing to locality, 
\begin{align*}
  \ell (A \ell (B))
  &= \bigcup\{\ell(x\ell(y)) \mid x\in A,\,  y\in B \text{ and }
    r(x) = \ell(\ell(y))\}\\
  &= \bigcup\{\ell(xy) \mid x\in A,\, y\in B \text{ and }
    r(x) = \ell(y)\}\\
 &= \ell (AB)
     \end{align*}
     and the opposite result is obvious.  
 \end{proof}

\end{document}